\documentclass[journal,twoside,web]{ieeecolor}
\usepackage[usenames,dvipsnames,table,xcdraw]{xcolor}
\usepackage{generic}
\usepackage{cite}
\usepackage{amsmath,amssymb,amsfonts, mathtools}
\usepackage{algorithmic}
\usepackage{graphicx}
\usepackage{algorithm,algorithmic}
\usepackage{hyperref}
\usepackage{textcomp}
\usepackage{paralist}
\newtheorem{definition}{\textbf{Definition}}
\newtheorem{proposition}{\textbf{Proposition}}
\newtheorem{remark}{\textbf{Remark}}

\newtheorem{lemma}{\textbf{Lemma}}
\newtheorem{assumption}{\textbf{Assumption}}
\newtheorem{theorem}{\textbf{Theorem}}

\allowdisplaybreaks
\newcommand{\vSGNE}{\mbox{v-SGNE}\ }
\newcommand{\Definition}{\mbox{\textit{Definition}}}
\newcommand{\Assumption}{\mbox{\textit{Assumption}}}

\newcommand{\Remark}{\mbox{\textit{Remark}}}

\newcommand{\SYG}[1]{\textcolor{black}{#1}}
\newcommand{\SYR}[1]{\textcolor{black}{#1}}
\newcommand{\SY}[1]{\textcolor{black}{#1}}

\begin{document}
\title{Stochastic Generalized Dynamic Games with Coupled Chance Constraints}
\author{Seyed Shahram Yadollahi, Hamed Kebriaei, and Sadegh Soudjani
}
\maketitle
\begin{abstract}
\SYG{This paper investigates stochastic generalized dynamic games with coupling chance constraints, where agents have incomplete information about uncertainties satisfying a concentration of measure property. This problem, in general, is non-convex and NP-hard. To address this, we propose a convex under-approximation by replacing chance constraints with tightened expected-value constraints, yielding a tractable game. We prove the existence of a stochastic generalized Nash equilibrium (SGNE) in this new game and show that its variational SGNE is an $\boldsymbol{\varepsilon}$-SGNE for the original game, with $\boldsymbol{\varepsilon}$ expressed via the approximation errors and Lagrange multipliers. A semi-decentralized, sampling-based algorithm with time-varying step sizes is developed, requiring no prior knowledge of the uncertainty distribution or expectation evaluations. Unlike existing methods, it avoids step-size tuning based on Lipschitz constants or adaptive rules. Under standard assumptions on the pseudo-gradient, the algorithm converges almost surely to an SGNE.
}
\end{abstract}
\section{Introduction}\label{sec: introduction}
Stochastic dynamic games \cite{Basar-zaccour-2018-handbook} have gained significant attention in recent decades \cite{Basar-zaccour-2018-handbook,song2024Decision}. Safety constraints, especially in multi-agent systems \cite{zhang2021physical}, pose additional challenges when agents face both local and coupling constraints. Analyzing such games requires the notion of a \emph{generalized Nash equilibrium} (GNE) \cite{facchinei2010generalized}, with applications in energy management \cite{yadollahi2023generalized} and autonomous vehicle control \cite{xie2023stochastic}. While GNE problems have been well studied in static and deterministic settings \cite{facchinei2010generalized,belgioioso2022distributed}, accounting for dynamics and uncertainties makes the problem more realistic yet substantially harder.

Stochastic generalized games have been studied in works such as \cite{ravat2011characterization,Franci-TAC-2021-Linear}. The former characterizes stochastic GNE (SGNE) via KKT conditions and stochastic variational inequalities but does not provide computational algorithms. \SYR{Subsequent works \cite{Franci-TAC-2021-Linear,franci2021stochastic,zou-Lygeros-ZOAlg_2023} developed semi-decentralized and distributed SGNE-seeking algorithms for settings where uncertainty affects \textit{only} objective functions, modeled as expected convex functions. These methods assume affine or separable convex coupling constraints and prove convergence under strong monotonicity or co-coercivity, with \cite{franci2021stochastic} relaxing this to monotonicity.} 
\SYR{The algorithm in \cite{franci2021stochastic} employs a fixed step size scaled by the inverse of an operator’s Lipschitz constant, whose computation can be challenging in nonlinear and stochastic settings.
}
Incorporating probabilistic guarantees has motivated the study of chance-constrained games. Works such as \cite{singh2019second} consider local chance constraints, while \cite{TomlinSojoudi2023} addresses coupling chance constraints using scenario-based methods and ADMM, but at high computational cost.
Dynamic extensions of stochastic generalized games (SGDGs) have also emerged. For example, \cite{reddy2019open} analyzed SGDGs with non-probabilistic affine coupling constraints, while \cite{zhong2022chance} proposed a two-stage algorithm for SGDGs with Gaussian uncertainty, approximating chance constraints linearly but without convergence guarantees. Vehicle safety applications have been studied using SGDGs with quadratic objectives and affine coupling chance constraints \cite{xie2023stochastic}. Energy management applications have also been explored: \cite{yadollahi2023generalized} applied SGDGs with quadratic costs and linear coupling chance constraints to microgrids with shared batteries, under assumptions on distributions’ support and moments, and proposed a deterministic reformulation.

Only a few works address dynamics or uncertainties in GNE analysis \cite{franci2021stochastic,Franci-TAC-2021-Linear,reddy2015open,yadollahi2023generalized}. The closest work to ours is \cite{franci2021stochastic}, which develops an SGNE-seeking algorithm under stochasticity in objectives. We extend this framework to SGDGs with stochastic dynamics and shared coupling chance constraints.
This paper studies a \emph{stochastic generalized dynamic game} (SGDG) with shared stochastic dynamics and coupling chance constraints. The dynamics are linear, time-varying, and influenced by players’ decisions and stochastic disturbances, whose distribution is unknown but belongs to a class satisfying a concentration of measure property. The SGDG is reformulated into a stochastic generalized Nash equilibrium (SGNE) problem, where the chance constraints are under-approximated by convex expected constraints. This yields an “under-approximated game” with expected costs and coupling constraints. We show that any SGNE of this game is also an $\varepsilon$-SGNE of the original SGDG. We further investigate the existence of such equilibria. Since the disturbance distribution is unknown and expectation computations are generally intractable, we propose a sample-based, semi-decentralized algorithm for computing an SGNE.
The main contributions are:
\begin{itemize}
\item First results on SGDG with coupling chance constraints under general uncertainty distributions satisfying concentration of measure, with more general cost and constraint structures than prior works \cite{franci2021stochastic,yadollahi2023generalized}. Unlike \cite{zhong2022chance,xie2023stochastic}, our method does not require full distributional knowledge or Gaussian assumptions.
\item Reformulation of the SGDG into an SGNE problem via convex under-approximation of chance constraints, and providing the proof that equilibria of the under-approximated game are $\varepsilon$-SGNEs of the original SGDG.
\item \SYR{We propose a semi-decentralized SGNE-seeking algorithm based on operator theory and stochastic approximation with a preset time-varying step size, eliminating the need for Lipschitz constant knowledge or adaptive estimation rules \cite{franci2021stochastic, malitsky2020golden}.
}
\item Existence results for SGNE and proof of almost sure convergence of the algorithm under pseudo-gradient \mbox{monotonicity}.
\end{itemize}
\textbf{Notations}:
Let $\mathbb{N} = 0,1,2,\ldots$ denote the natural numbers, and $\mathbb{R}^m$ ($\mathbb{R}^m_{\geq 0}$) the (nonnegative) $m$-dimensional Euclidean space. For $x \in \mathbb{R}^m$, $\boldsymbol{A} \in \mathbb{R}^{m \times n}$, let $x^\top$, $\boldsymbol{A}^\top$ denote transposes. The inner product is $\langle x, y \rangle = x^\top y$, with norm $\|x\|_2 = \sqrt{x^\top x}$; the Frobenius norm is $\|\boldsymbol{A}\|_F$. 
For vectors \(x_i\), define \(\mathrm{col}(x_i)_{i=1}^{N} := [x_1^\top, \ldots, x_N^\top]^\top\), and let \((x_i)_{i=1}^{N}\) denote their sequence.
\section{Modeling Framework and Problem Statement}\label{sec: Modeling Framework and Problem Statement}
We consider a Stochastic Generalized Dynamic Game (SGDG) in discrete time with coupling chance constraints. 
Let the set of players be denoted by
\(
\mathcal{I} := \{1, 2, \dots, N\}
\), and the finite time horizon by \( T \).
At each time step \( t \in \{0,1,\ldots,T-1\} \), each player \( i \in \mathcal{I} \) selects a control action
\(
u^i_t \in \mathcal{D}^i_t \subseteq \mathbb{R}^{n_i},
\)where \( \mathcal{D}^i_t \) is a nonempty, compact, and convex set.
For each player \( i \), the full sequence of strategies over the horizon is given by
\(
u^i = \mathrm{col}(u^i_t)_{t=0}^{T-1} \in \mathcal{D}^i := \prod_{t=0}^{T-1} \mathcal{D}^i_t
\). The collective strategy profile of all players is
\(
u = \mathrm{col}(u^i)_{i=1}^{N} \in \mathcal{D} := \prod_{j=1}^N \mathcal{D}^j
\), and the strategy profile excluding player \( i \) is
\(
u^{-i} = \mathrm{col}(u^{j})_{j=1,j\neq i}^{N} \in \mathcal{D}^{-i} := \prod_{j \ne i} \mathcal{D}^j
\). The players' strategies affect the global state of the system denoted by \( s_t \in \mathbb{R}^{n_{s}} \) through the following stochastic difference equation:
\begin{equation}
    \label{eq: pure dynamic}
    s_{t+1} = \boldsymbol{A}_t s_t + \sum_{j=1}^N \boldsymbol{B}^j_t u^j_t + w_t,\,\,\forall t \in \{0,1,\ldots,T-1\},
\end{equation}
where \( \boldsymbol{A}_t \in \mathbb{R}^{n_{s} \times n_{s}} \) and \( \boldsymbol{B}^j_t \in \mathbb{R}^{n_{s} \times n_j} \) are time-dependent system matrices, \( w_t \in \mathbb{R}^{n_{s}} \) represents the \SYR{random} disturbance at time \( t \), \( w := \mathrm{col}\left(w_t \right)_{t=0}^{T-1}\), and \( s_0 \) is the initial state which is assumed to be known. The state profile is $s=\mathrm{col}(s_t)_{t=0}^{T}\in\mathbb{R}^{(T+1)n_s}$. \SYR{ 
The adopted information structure is open-loop: Each player commits at the initial time to a fixed profile of actions, meaning its actions at all time instants are predetermined \cite[sec. 2.2]{Basar-zaccour-2018-handbook}.
This paper further assumes that the distributions of the random variables are only partially known.
Each player \(i \in \mathcal{I}\) aims to minimize its own cost function
\begin{equation}
\label{eq: pure_cost_function}
\mathbb{J}^i(s_{0}, u^i, u^{-i}) =
\mathbb{E}_{w}\!\left[\hat{J}^{i}(s) + \tilde{J}^{i}(u^i, u^{-i})\right],
\end{equation}
by selecting an appropriate strategy \(u^i\), where \(\hat{J}^{i}\!:\mathbb{R}^{(T+1)n_{s}}\!\to\!\mathbb{R}\) is the state-dependent cost and
\(\tilde{J}^{i}\!:\mathbb{R}^{T\sum_j n_j}\!\to\!\mathbb{R}\) is direct interaction costs.}
All players must satisfy the shared safety chance constraints
\begin{equation}
\label{eq: chance constraint}
    \hspace{-0.2cm}\mathbb{P}\left\{\hat{\xi}^{j}(s) + \tilde{\xi}^{j}(u) \leq 0 \right\} \geq 1 - \gamma^j, \quad j\in \SYG{\mathcal{S} := \{1, \dots, m\}},
\end{equation}
where \( \hat{\xi}^{j}: \mathbb{R}^{(T+1)n_{s}} \rightarrow \mathbb{R} \) and \( \tilde{\xi}^{j}: \mathbb{R}^{T \sum_j n_j} \rightarrow \mathbb{R} \).
Here, \( \hat{\xi}^{j}(s) \), \( \forall j \in \mathcal{S} \), is a measurable function, and $\gamma^j \in \left(0,1\right)$ 
is a constant-violation tolerance for the $j^{\text{th}}$ chance constraint.
\SYG{\begin{assumption}
\label{ass: cost function in system model}
For each $i \in \mathcal{I}$, the cost $\mathbb{J}^{i}$ is well-defined; $\hat{J}^{i}(s)$ and $\tilde{J}^{i}(u^{i},u^{-i})$ are convex differentiable in $s$ and $u^{i}$, respectively.
For each $j \in \SYG{\mathcal{S}}$, $\hat{\xi}^{j}(s)$ and $\tilde{\xi}^{j}(u^{i},u^{-i})$ are likewise convex differentiable in $s$ and $u^{i}$, respectively.
\end{assumption}}
\SYR{
Using the recursion for the system state from \eqref{eq: pure dynamic}, the profile 
\(s\) can be compactly expressed as
\(
s=\boldsymbol{\Theta}s_0+\sum_{j=1}^{N}\boldsymbol{\Gamma}^j u^j+\boldsymbol{\Upsilon}w,
\)
with appropriate matrices \(\boldsymbol{\Theta}\),
\(\boldsymbol{\Gamma}^j\),
and \(\boldsymbol{\Upsilon}\).
}
\SYG{Given the open-loop information structure, by substituting \( s \) with its equivalent expression into \( \hat{J}^{i}(s) \) in \eqref{eq: pure_cost_function} and \( \hat{\xi}^{j}(s) \) in \eqref{eq: chance constraint}, the introduced game can be reformulated as an \emph{equivalent} stochastic generalized game, denoted by \( \mathcal{G}_{1} \), where each agent \( i \in \mathcal{I} \) aim to choose a strategy, $u^{i}$, given the decision variables of the other agents, $u^{-i}$, and \textit{known parameter} $s_{0}$, to solves the following local optimization problem
\begin{align}\label{eq: each optimization of the reformulated game}
\begin{cases}
\hspace{0cm}\underset{u^i \in \mathcal{D}^{i}}{\min}\hspace{0.3cm} &\mathbb{E}_{w}\left[J^i \left(s_0, u^i, u^{-i}, w \right)\right]\\
\hspace{0cm}\textrm{s.t.} &\mathbb{P} \left\{ \bar{\xi}^j(s_0, u, w) \leq 0, \right\} \geq 1 - \gamma^j,\quad \forall j\in \mathcal{S}
\end{cases}
\end{align}
with \(J^i (s_0, u^i, u^{-i}, w ) = \hat{J}^{i} \big( \boldsymbol{\Theta} s_0 + \sum_{j=1}^{N} \boldsymbol{\Gamma}^j u^j + \boldsymbol{\Upsilon} w \big) 
+ \tilde{J}^{i}(u^i, u^{-i})\) and \( \bar{\xi}^j(s_0, u, w) = \hat{\xi}^{j} \big( \boldsymbol{\Theta} s_0 + \sum_{j=1}^{N} \boldsymbol{\Gamma}^j u^j + \boldsymbol{\Upsilon} w \big) + \tilde{\xi}^{j}(u) \), where both the cost functions and feasible strategy sets depend on the strategies of other players and on random variables. 
}
\SYG{
\begin{remark}\label{remark: convexity}
Under \Assumption~\ref{ass: cost function in system model}, since \( s \) is affine in \( u^i \) and expectation preserves convexity \cite[subsec. 3.2.1]{boyd2004convex}, \( J^i \), \( \mathbb{J}^{i} \), and \( \bar{\xi}^j \) are convex in \( u^i \), which is a standard requirement in chance-constrained optimisation/game-theoretic works\cite{soudjani2018concentration,TomlinSojoudi2023}.
\end{remark}}
The collective feasible set of strategies in \( \mathcal{G}_{1} \) can be defined as $\mathcal{U}_{\mathcal{G}_{1}}(s_0) =  \{u \in \mathcal{D} \mid \forall j\in \SYG{\mathcal{S}},\: \mathbb{P} \left\{  \bar{\xi}^j(s_0, u, w) \leq 0, \right\} \geq \mbox1 - \gamma^j\}$, and the feasible set for player \(i\) is as
$ \mathcal{U}_{\mathcal{G}_{1}}^{i} \!(s_{0},u^{-i}) \!=\! \{u^{i} \in \mathcal{D}^{i} \!\!\mid\! \exists u^{-i}, (s_{0},u^{i},u^{-i}) \in \mathcal{U}_{\mathcal{G}_{1}}(s_{0})\}
$.
\begin{definition} 
\label{def: varepsilon SGNE}
An \(\varepsilon-\)SGNE for \( \mathcal{G}_{1} \) is a strategy \( u_{\star} \in \mathcal{U}_{\mathcal{G}_{1}}\) such that for all \(i \in \mathcal{I}\),
\[
    \mathbb{J}^{i}(s_{0},u_{\star}^{i},u_{\star}^{-i}) \leq \inf \big\{ \mathbb{J}^{i}(s_{0},y,u_{\star}^{-i}) \mid y \in \mathcal{U}_{\mathcal{G}_{1}}^{i}(s_{0},u^{-i}_{\star}) \big\} + \varepsilon.\]
\end{definition} 

An \(\varepsilon-\)SGNE represents a strategy profile where no player can reduce their cost by more than \(\varepsilon\) through unilateral deviation. If \Definition~\ref{def: varepsilon SGNE} holds with $\varepsilon = 0$, then $u_{\star}$ is a SGNE for $\mathcal{G}_{1}$. In this paper, we aim to investigate $\varepsilon-$SGNE in $\mathcal{G}_{1}$. 
\section{Convex Under-Approximation of the SGDG} \label{sec: under approximated game}
Solving the original stochastic game $\mathcal{G}_{1}$ with non-convex chance constraints is NP-hard and generally intractable \cite{conitzer2008new}. We therefore propose a convex under-approximation based on a key assumption about the uncertainty's distribution.
\begin{assumption}[Concentration of Measure (CoM)]
 \label{ass: concentration of measure}
 The random variable \( w \) in \eqref{eq: each optimization of the reformulated game} is defined on the probability space \( \left( \Xi_{w}, \mathcal{F}_{w}, \mathbb{P}_{w} \right) \)  satisfying the inequality
        $\mathbb{P}_{w}\left\{|\phi( w)-\mathbb{E}_{w}[\phi(w)]|\leq \theta \right\} \geq 1- h(\theta),\forall \theta\geq 0$,
    for any Lipschitz continuous function \( \phi : \Xi_{w} \to \mathbb{R} \). Here, \( h : \mathbb{R}_{\geq 0} \to [0, 1] \) is a monotonically decreasing function, and for each \( j \in \SYG{\mathcal{S}} \), \( \bar{\xi}^{j}(s_0, u, \cdot) \) is  Lipschitz continuous. 
 \end{assumption}
The CoM phenomenon states that in a probability space, if a set has a measure of at least one-half, most points are close to it. Additionally, if a function on this space is sufficiently regular, the probability of it deviating significantly from its expectation (median) is low \cite{soudjani2018concentration}.
Assumption~\ref{ass: concentration of measure} is especially relevant in high-dimensional spaces, where probability tends to concentrate in small regions. Many distributions, including Gaussian and log-concave, have this property \cite{soudjani2018concentration}. For more details about the CoM (such as the form of $h(\cdot)$ for different distributions), we refer readers to \cite{soudjani2018concentration} and its references.

According to  \cite[Proposition 4]{soudjani2018concentration} and \Assumption~\ref{ass: concentration of measure}, the feasible domain of the chance constraints in \eqref{eq: each optimization of the reformulated game} includes the compact expected constraint
\(\mathbb{E}_{w}\left[g(s
_{0},u,w)\right] \leq 0\),
where \( g(s_{0},u,w) := \bar{\xi}(s_{0},u,w) + h^{-1}(\gamma) + \beta \) with \( \bar{\xi}(s_{0},u,w) = \mathrm{col}(\bar{\xi}^{j}(s_{0},u,w))_{j=1}^{m}\), \( \gamma := \mathrm{col}(\gamma^{j})_{j=1}^{m} \), and \( h^{-1}(\gamma) = \mathrm{col}(h^{-1}(\gamma^{j}))_{j=1}^{m} \). The vector \(\beta = \mathrm{col}\left(\beta^{j}\right)_{j=1}^{m}\), for all $\beta \geq 0$, can be empirically adjusted to control the tightness level of the under-approximation. 
\SYR{\SYG{Furthermore}, based on \Assumption~\ref{ass: cost function in system model} and similar to \Remark~\ref{remark: convexity}, it can be easily shown that this expected constraint represents certain convex constraints.}\\
\SYG{Based on the above discussion, we introduce a new \emph{stochastic generalized game}, denoted by $\mathcal{G}_{2}$. This game has a similar structure to $\mathcal{G}_{1}$, except that each player $i \in \mathcal{I}$ solves the following optimization problem:
}
\SYG{
\begin{align}\label{eq: optimization of under-approximated game}
\begin{cases}
\hspace{0cm}\underset{u^i \in \mathcal{D}^{i}}{\min}\hspace{0.3cm} &\mathbb{E}_{w}\left[J^i \left(s_0, u^i, u^{-i}, w \right)\right]\\
\hspace{0cm}\textrm{s.t.} &\mathbb{E}_{w}\left[g(s
_{0},u,w)\right] \leq 0.
\end{cases}
\end{align}}
\SYG{The \emph{collective} feasible strategy set for \( \mathcal{G}_{2} \), denoted by \( \mathcal{U}_{\mathcal{G}_{2}}(s_0) \), and the feasible strategy set for player \( i \), denoted by \( \mathcal{U}_{\mathcal{G}_{2}}^{i}(s_0, u^{-i}) \), are defined analogously to those in \( \mathcal{G}_{1} \).}
Note that $\mathcal{U}_{\mathcal{G}_{2}}(s_0)\subseteq \mathcal{U}_{\mathcal{G}_{1}}(s_0)$ and $\mathcal{U}_{\mathcal{G}_{2}}^{i}(s_0,u^{-i})\subseteq \mathcal{U}_{\mathcal{G}_{1}}^{i}(s_0,u^{-i})$.
\SYG{Although $\mathcal{G}_{2}$ is a convex \mbox{under-approximated} game for $\mathcal{G}_{1}$, the expectations in \eqref{eq: optimization of under-approximated game} cannot, in general, be exactly computed or expressed in closed (deterministic) form, since the \mbox{distributions} of the random variables are not fully known. Thus $\mathcal{G}_{2}$ is a stochastic game.
}
\begin{assumption}
\label{ass: second assumption about cost function}
The set \( \mathcal{U}_{\mathcal{G}_{2}} \) satisfies the Slater's constraint qualification. For each \( s_0 \in \mathbb{R}^{n_{s}} \) and \( w \in \Xi_{w} \), \( g(s_0, \cdot, w) \) is \( \ell_{g}(s_0, w) \)-Lipschitz continuous. \mbox{Additionally}, 
\( g(s_0, u, w) \)
and \( \nabla_{u} g(s_0, u, w) \) are bounded, meaning that \( \sup_{u \in \mathcal{U}_{\mathcal{G}_{2}}} \| g(s_0, u, w) \| \leq B_{1g}(s_0, w) \) and \( \sup_{u \in \mathcal{U}_{\mathcal{G}_{2}}} \| \nabla_{u} g(s_0, u, w) \| \leq B_{2g}(s_0, w) \). 
For each \( i \in \mathcal{I} \), \( s_0 \in \mathbb{R}^{n_{s}} \), \( u^{-i} \in \mathcal{D}^{-i} \), and \( w \in \Xi_{w} \), \( J^{i}(s_0, \cdot, u^{-i}, w) \) in \eqref{eq: optimization of under-approximated game} is \( \ell_{J^{i}}(s_0, u^{-i}, w) \)-Lipschitz continuous. 
The constants \( \ell_{g}(s_0, w) \), \( B_{1g}(s_0, w) \), \( B_{2g}(s_0, w) \)
, and \( \ell_{J^{i}}(s_0, u^{-i}, w) \), for all $i \in \mathcal{I}$, are integrable with respect to \( w \).
\end{assumption}

The SGNE of $\mathcal{G}_{2}$ is also defined according to \Definition~\ref{def: varepsilon SGNE}.
In the following, we show that there exists an SGNE of \( \mathcal{G}_{2} \), which is an \(\varepsilon\)-SGNE for \( \mathcal{G}_{1} \).

\subsection{Characterization of SGNE in \texorpdfstring{$\mathcal{G}_{2}$}{G2}}
We focus on a key subclass of SGNE in \(\mathcal{G}_{2}\), which \mbox{corresponds} to the solution set for \(u_\star\) in the following \mbox{stochastic} variational inequality (SVI): 
\begin{equation}
\label{eq: SVI}
    \langle \mathbb{F}_{\mathcal{G}_{2}}\left(s_{0},u_{\star}\right),u-u_{\star}\rangle \geq 0,\,\,  \forall u \in \mathcal{U}_{\mathcal{G}_{2}}(s_0),
\end{equation}
where $\mathbb{F}_{\mathcal{G}_{2}}$ is the pseudo-gradient mapping defined as
$\mathbb{F}_{\mathcal{G}_{2}}\left(s_{0},u\right) := \mathbb{E}_{w}\left[F_{\mathcal{G}_{2}}(s_{0},u^{i},u^{-i},w)\right]$, 
where $F_{\mathcal{G}_{2}}\left(s_{0},u^{i},u^{-i},w\right) = \mathrm{col}\left(\nabla_{u^{i}}J^{i}\left(s_{0},u^{i},u^{-i},w\right)\right)_{i=1}^{N}$.
Under \textit{Assumptions~\ref{ass: cost function in system model}--\ref{ass: second assumption about cost function}}, it follows from \cite[Proposition 12.7]{palomar2010convex} that any solution of $\mathrm{SVI}(\mathcal{U}_{\mathcal{G}_{2}},\mathbb{F}_{\mathcal{G}_{2}})$ in \eqref{eq: SVI} is an SGNE for $\mathcal{G}_{2}$ while vice versa does not hold in general.
\begin{proposition}\label{proposition: existence of vSGNE G3}
    If \textit{Assumptions~\ref{ass: cost function in system model}--\ref{ass: second
    assumption about cost function}} hold, the solution set of $\mathrm{SVI}(\mathcal{U}_{\mathcal{G}_{2}},\mathbb{F}_{\mathcal{G}_{2}})$ is not empty \cite[Corollary 2.2.5]{facchinei2003finite}.
\end{proposition} 

We call \emph{variational SGNE} (\vSGNE) those SGNE that are also solutions of the associated SVI, namely the solution of $\mathrm{SVI}(\mathcal{U}_{\mathcal{G}_{2}},\mathbb{F}_{\mathcal{G}_{2}})$ in \eqref{eq: SVI}. 
Given that \(\mathcal{U}_{\mathcal{G}_{2}}\) meets the Slater constraint qualification as stated in \textit{Assumption~\ref{ass: second assumption about cost function}}, based on  \cite[Proposition 1.3.4]{facchinei2003finite}, \( u_{\star} \) is a \vSGNE if and only if there exists a \( \overline{\lambda} \in \mathbb{R}^{m}_{\geq 0} \) such that the following KKT inclusions hold for any \( i \in \mathcal{I} \):
\begin{equation}
\label{eq: KKT conditions associated with each agent in vSGNE}
    \begin{cases}
        0 \in \mathbb{E}_{w}\left[\nabla_{u^{i}}J^{i}(s_{0},u_{\star}^{i},u_{\star}^{-i},w)\right] + \mathcal{N}_{\mathcal{D}^{i}}(u_{\star}^{i})
        \\ 
        \hspace{0.6cm} + \mathbb{E}_{w}\left[\nabla_{u^{i}}g(s_{0},u_{\star}^{i},u_{\star}^{-i},w)\right] \overline{\lambda}, \\
        0 \in -\mathbb{E}_{w}\left[g(s_{0},u_{\star}^{i},u_{\star}^{-i},w)\right] + \mathcal{N}_{\mathbb{R}^{m}_{\geq 0}}(\overline{\lambda}).
    \end{cases}
\end{equation}
Here normal cone operator $\mathcal{N}_{\mathcal{U}}(x)=\{v\mid \langle v,y-x\rangle\le0,\ \forall y\in\mathcal{U}\}$. \SYR{The interchangeability of the expected value and the gradient in $\mathbb{F}_{\mathcal{G}_{2}}$ and  \eqref{eq: KKT conditions associated with each agent in vSGNE} are guaranteed under \textit{Assumptions~\ref{ass: cost function in system model}}
and \textit{\ref{ass: second assumption about cost function}} \cite[Lemma 3.4]{ravat2011characterization}}. We recast the KKT conditions in \eqref{eq: KKT conditions associated with each agent in vSGNE} as a compact operator inclusion:
\begin{align}
\label{eq: KKT in operator and compact form}
    0 \in \mathcal{T}(u,\overline{\lambda}) := \mathcal{A}\left(u,\overline{\lambda}\right)+\mathcal{B}\left(u,\overline{\lambda}\right),
\end{align}
where
\begin{align}
\label{eq: mathcal_A and mathcal_B in centralized algorithm}
    \mathcal{A}:\begin{bmatrix}
        u\\
        \overline{\lambda}
    \end{bmatrix} &\mapsto \begin{bmatrix}
        \mathbb{F}_{\mathcal{G}_{2}}(s_{0},u)\\
        0
    \end{bmatrix} + \begin{bmatrix}
        \mathbb{E}_{w}\left[\nabla_{u}g(s_{0},u,w)\right]\overline{\lambda}\\
        -\mathbb{E}_{w}\left[g(s_{0},u,w)\right]
    \end{bmatrix}, \nonumber \\
    \mathcal{B}:\begin{bmatrix}
        u\\
        \overline{\lambda}
    \end{bmatrix} &\mapsto \begin{bmatrix}
        \mathcal{N}_{\mathcal{D}}(u)\\
        \mathcal{N}_{\mathbb{R}^{m}_{\geq 0}}(\overline{\lambda})
    \end{bmatrix}.
\end{align}
Thus, \( u_{\star} \) is a \vSGNE if and only if there exists a \( \overline{\lambda}_{\star} \in \mathbb{R}^{m}_{\geq 0} \) such that \( \mathrm{col}(u_{\star},\overline{\lambda}_{\star}) \in \mathrm{zer}(\mathcal{A} + \mathcal{B}) \).
\begin{assumption}
\label{ass: Monotonicity and Lipschitz of mathcal_F}
    The mapping \( \mathbb{F}_{\mathcal{G}_{2}} \) is monotone and \( \ell_{\mathbb{F}_{\mathcal{G}_{2}}} \)-Lipschitz continuous for some \( \ell_{\mathbb{F}_{\mathcal{G}_{2}}} > 0 \).
\end{assumption} 
\SYR{\begin{lemma}
\label{lem: equivalence of answer distributed and centralized algorithm (KKT)}
    Let \textit{Assumptions~\ref{ass: cost function in system model}--\ref{ass: Monotonicity and Lipschitz of mathcal_F}},
    hold. Then, we have: 
    \begin{inparaenum}[(i)]
        \item \( \mathrm{zer}(\mathcal{A} + \mathcal{B}) \neq \emptyset \),
        \item \( \mathcal{A} \) is monotone and \( \ell_{\mathcal{A}} \)-Lipschitz continuous, and \label{lem: monotonicity of A and Lipschitz}
        \item  \( \mathcal{B} \) is maximally monotone.
    \end{inparaenum}\\
    \begin{proof}
    The proof is similar to \cite[Lemma 1, 2]{franci2021stochastic} and \cite{ravat2011characterization}.
\end{proof}
\end{lemma}
}
\subsection{SGNE of \texorpdfstring{$\mathcal{G}_{2}$}{G2} as an \texorpdfstring{$\varepsilon$}{epsilon}-SGNE of \texorpdfstring{$\mathcal{G}_{1}$}{G1}}
In this subsection, we focus on establishing a relation between the solution of the under-approximated game $\mathcal{G}_{2}$ and that of the actual game $\mathcal{G}_{1}$. 
We demonstrate that the \vSGNE of $\mathcal{G}_{2}$ is an \(\varepsilon-\)SGNE for \(\mathcal{G}_1\) while \(\varepsilon\) is bounded with a duality-like gap induced by the under-approximation of the feasible set of \(\mathcal{G}_1\).
\SYG{
\begin{theorem}\label{theorem: Epsolon-newVersion}
\SYG{Let us define  \(
        M_{i}^{j} := \mathrm{sup}_{u^{i} \in \mathcal{U}_{\mathcal{G}_{1}}^{i}} \|1 - \gamma^{j}- \mathbb{P}\left\{\bar{\xi}^{j}(s_{0},u^{i},u_{\star}^{-i},w)\leq 0 \right\}- \mathbb{E}_{w}\left[g^{j}(s_0,u^i,u_{\star}^{-i},w)\right]\|,
\)
}
for all $j \in \SYG{\mathcal{S}}$ and $i \in \mathcal{I}$, and let \textit{Assumptions \ref{ass: cost function in system model}--\ref{ass: second assumption about cost function}} hold. If \( u_{\star} \) is a \vSGNE of the game \( \mathcal{G}_{2} \), and $\overline{\lambda}_{\SYG{\star}}$ is its corresponding Lagrangian multiplier. Then, \( u_{\star} \) is an \( \varepsilon \)-SGNE for the game \( \mathcal{G}_1 \) with \( \varepsilon  = \max \{\varepsilon^{i} \mid i \in \mathcal{I}\}\), such that for each $i \in \mathcal{I}$, $\varepsilon^{i}$ is as  $\varepsilon^{i} = \sum_{j=1}^{m}\overline{\lambda}_{\star}^{j} M_{i}^{j}$.
\end{theorem}}
\begin{proof}
    Let $\overline{\lambda}_{\star} = \mathrm{col}\left(\overline{\lambda}_{\star}^{j}\right)_{j=1}^{m} \in \mathbb{R}^{m}_{\geq 0}$. For each $i \in \mathcal{I}$ and $u^{i} \in \mathcal{U}_{\mathcal{G}_{1}}^{i}\left(s_{0},u^{-i}_{\star}\right)$, we have
    \begin{align*}
        &\mathbb{J}^{i}(s_{0},u^{i},u_{\star}^{-i}) \geq  \mathbb{J}^{i}(s_{0},u^{i},u_{\star}^{-i}) +  \\
        &+\sum_{j=1}^{m} \Bigg\{ \overline{\lambda}_{\star}^{j}\Bigg( \mathbb{E}_{w}\left[g^{j}(s_0,u^i,u_{\star}^{-i},w)\right] - \mathbb{E}_{w}\left[g^{j}(s_0,u^i,u_{\star}^{-i},w)\right] \\
        & + \underbrace{1 - \gamma^{j}- \mathbb{P}\left\{\bar{\xi}^{j}(s_{0},u^{i},u_{\star}^{-i},w)\leq 0 \right\}}_{\leq 0} \Bigg)\Bigg\} \\
        &\geq \mathbb{J}^{i}(s_{0},u^{i},u_{\star}^{-i}) +  \sum_{j=1}^{m} \overline{\lambda}_{\star}^{j} \Bigg(\mathbb{E}_{w}\left[g^{j}(s_0,u^i,u_{\star}^{-i},w)\right] - M_{i}^{j} \Bigg).
    \end{align*}
    Multiplying both sides of the above inequality by a minus sign and subsequently adding \(\mathbb{J}^{i}(s_{0}, u_{\star}^{i}, u_{\star}^{-i})\) to the both sides, \SYR{given that $\mathcal{U}_{\mathcal{G}_{1}}^{i}\left(s_{0},u^{-i}_{\star}\right) \subseteq \mathcal{D}^{i}$, we obtain}
    \SYR{\begin{align}
    \label{eq: second equation of epsilon proof}
        &\mathbb{J}^{i}(s_{0}, u_{\star}^{i}, u_{\star}^{-i}) - \mathbb{J}^{i}\left(s_{0},u^{i},u_{\star}^{-i}\right) \nonumber \\ 
        &\leq \mathbb{J}^{i}(s_{0}, u_{\star}^{i}, u_{\star}^{-i})+\sum_{j=1}^{m}\overline{\lambda}_{\star}^{j} M_{i}^{j} -  \nonumber \\
        &\hspace{0.3cm}\Bigg(\mathbb{J}^{i}(s_{0},u^{i},u_{\star}^{-i}) + \sum_{j=1}^{m} \overline{\lambda}_{\star}^{j} \mathbb{E}_{w}\left[g^{j}(s_0,u^i,u_{\star}^{-i},w)\right]\Bigg) \nonumber \\
        & \leq \mathbb{J}^{i}(s_{0}, u_{\star}^{i}, u_{\star}^{-i})+\sum_{j=1}^{m}\overline{\lambda}_{\star}^{j} M_{i}^{j} -  \text{inf}_{u^{i} \in \mathcal{D}^{i}} \Big[\mathbb{J}^{i}(s_{0},u^{i},u_{\star}^{-i}) \nonumber\\
        &
        \quad+ \sum_{j=1}^{m} \overline{\lambda}_{\star}^{j} \mathbb{E}_{w}\left[g^{j}(s_0,u^i,u_{\star}^{-i},w)\right]\Big] \leq \sum_{j=1}^{m}\overline{\lambda}_{\star}^{j} M_{i}^{j}.
    \end{align}}  
    \SYR{where the last inequality is deduced from the strong duality property of the players' best response in $\mathcal{G}_{2}$.} Finally, setting the right-hand side of \eqref{eq: second equation of epsilon proof} to $\varepsilon^{i}$ and letting $\varepsilon = \max\{\varepsilon^1, \dots, \varepsilon^N\} $,
    we get that $u_{\star}$ is a $\varepsilon-$SGNE for $\mathcal{G}_{1}$.
\end{proof}
\begin{remark}\label{remark: epsilon interpretation}
    \SYG{If \( M^{j}_{i} = 0 \)} for all \( j \in \SYG{\mathcal{S}} \) and $i \in \mathcal{I}$, then we can conclude that \( \mathcal{U}_{\mathcal{G}_{1}}^{i}(s_{0},u^{-i}_{\star}) = \mathcal{U}_{\mathcal{G}_{2}}^{i}(s_{0},u^{-i}_{\star}) \) for all \( i \in \mathcal{I} \). \SYG{Moreover, $\overline{\lambda}_{\star}^{j}$ can be interpreted as the maximum shadow price of under-approximating the $j^{\mathrm{th}}$ constraint of $\mathcal{G}_{1}$; and $\varepsilon$ as the maximum cost of under-approximation, which can be imposed on a player for the result of $\mathcal{G}_{2}$. In Section~\ref{sec: Proposed SGNE Seeking Algorithm}, we propose an algorithm for finding $u_{\star}$ and $\overline{\lambda}_{\star}$.} 
\end{remark}
\section[SGNE Seeking Algorithm for G2]{SGNE Seeking Algorithm for \texorpdfstring{$\mathcal{G}_{2}$}{G2}}
\label{sec: Proposed SGNE Seeking Algorithm}
In this section, we propose \SYR{a modified stochastic version of the
Golden Ratio Algorithm\cite{malitsky2020golden}} to obtain the \vSGNE of $\mathcal{G}_{2}$, which corresponds to the zeros of the operator \( \mathcal{T} = \mathcal{A}+\mathcal{B} \) as shown in \eqref{eq: KKT in operator and compact form}--\eqref{eq: mathcal_A and mathcal_B in centralized algorithm}. 
Since computing the expected-value mappings in \( \mathcal{A} \) requires knowledge of the distributions of the random variables, which are unknown, we approximate \( \mathcal{A} \) by
\begin{align}
    \widehat{\mathcal{A}}:
    \begin{bmatrix}
        (u, \boldsymbol{w})\\
        \overline{\lambda}
    \end{bmatrix} \mapsto
    &\begin{bmatrix}
        \widehat{F}_{\mathcal{G}_{2}}(s_{0}, u, \boldsymbol{w})\\
        0
    \end{bmatrix} + 
    \begin{bmatrix}
        \widehat{H}(s_{0}, u, \boldsymbol{w}) \overline{\lambda}\\
        -\widehat{G}(s_{0}, u, \boldsymbol{w})
    \end{bmatrix},
\end{align}
where \( \widehat{F}_{\mathcal{G}_{2}}\), \(\widehat{H}\), and \(\widehat{G} \) are approximations of \( \mathbb{F}_{\mathcal{G}_{2}}\), \(\mathbb{H}\), and \(\mathbb{E}_{w}[g] \) using samples \( \boldsymbol{w} \), respectively. Here,
$
\mathbb{H}(s_{0}, u) = \mathbb{E}_{w}\left[ H(s_{0}, u,w) \right],
$
with
$
H(s_{0}, u,w) =\text{col}\big( H^{i}(s_{0}, u, w)\big)_{i=1}^{N} := \text{col}\big( \nabla_{u^{i}}g(s_{0}, u^{i}, u^{-i},w)\big)_{i=1}^{N}$. 
For brevity, we omit $s_0$ from the arguments when not essential.
We propose Algorithm \ref{alg: suggested algorithm} to compute a \vSGNE of \( \mathcal{G}_{2} \) in a \SYR{sampling-based} semi-decentralized manner. At iteration \( k \):\\
\begin{inparaenum}
    \item \textit{Coordinator Step:} The coordinator collects \( u^{i}_{(k)} \) from all players, draws samples $\boldsymbol{w}^{0}_{(k),[1]}, \dots, \boldsymbol{w}^{0}_{(k),[M_{(k)}]}$, and approximates \( \mathbb{E}_{w}[g(s_{0}, u_{(k)}, w)] \) (see \eqref{eq: approximation of Eg at each iteration}). It then computes \( \tilde{\lambda}_{(k)} \in \mathbb{R}^{m} \) via an averaging step with inertia (see \eqref{eq: averaging step lambda that computing lambda tilde at each iteration}), updates \( \overline{\lambda}_{(k+1)} \) by projection step \eqref{eq: compute lambda at each iteration}, and broadcasts it to all players.
    \item \textit{Players Step:} each player \( i \) uses local samples $ \boldsymbol{w}^{i}_{(k),[1]}, \dots, \boldsymbol{w}^{i}_{(k),[M_{(k)}]}$ to approximate \( \mathbb{F}_{\mathcal{G}_{2}}^{i} \) and \( \mathbb{H}^{i} \) (see \eqref{eq: computing F hat i, at each iteration k}, \eqref{eq: compute Lambda hat i at each iteration k}), computes \( \tilde{u}^{i}_{(k)} \) via averaging step \eqref{eq: averaging step u that compute tilde u}, and then updates \( u^{i} \) through a projection step \eqref{eq: projection step u}.
    The updated values are then sent back to the coordinator and other players.
\end{inparaenum}\\ 
    In Algorithm~\ref{alg: suggested algorithm}, $\boldsymbol{w}^{i}_{\left(k\right)} = \big(\boldsymbol{w}^{i}_{\left(k\right),[l]}\big)_{l=1}^{M_{(k)}}$, for $i=0,1,\cdots,N$ and $k \in \mathbb{N}$.
    The samples \( \boldsymbol{w}^{i}_{(k),[l]} \) are drawn independently from \( \mathbb{P}_{w} \), and the step-size \( \alpha_{(k)} \) is positive.
\begin{algorithm}
\caption{Semi-decentralized sampling-based computation of the \vSGNE of $\mathcal G_2$.}
\begin{algorithmic}[1]
\STATE \textbf{Initialization:} $u^{i}_{(0)}, \tilde{u}^{i}_{(-1)} \in \mathcal{D}^{i}$, and $  \overline{\lambda}_{(0)},  \tilde{\lambda}_{(-1)} \in \mathbb{R}_{\geq 0}^m$, $M_{(k)}$, $\delta$, $\alpha_{(k)}$.
\STATE \textbf{Iteration $k$:} \\
\: (1) \textbf{Coordinator}: Receive $u^{i}_{(k)}$ for all $i \in \mathcal{I}$, and  update:
\begin{equation}
\label{eq: approximation of Eg at each iteration}
\hspace{-0.7cm}\widehat{G}(s_{0}, u_{(k)}, \boldsymbol{w}^{0}_{(k)}) = \frac{\sum_{l=1}^{M_{(k)}} g\left(s_{0}, u_{(k)}, \SYR{\boldsymbol{w}^{0}_{(k),[l]}}\right)}{M_{(k)}},    
\end{equation}
\begin{equation}
\label{eq: averaging step lambda that computing lambda tilde at each iteration}
\tilde{\lambda}_{(k)} = (1-\delta)\overline{\lambda}_{(k)} + \delta \tilde{\lambda}_{(k-1)},   
\end{equation}
\begin{equation}
\label{eq: compute lambda at each iteration}
\hspace{-0.3cm}\overline{\lambda}_{(k+1)} = \text{proj}_{\mathbb{R}_{\geq 0}^m}\bigg\{\tilde{\lambda}_{(k)} + \alpha_{\left(k\right)} \hat{G}\left(s_{0},u_{(k)},\boldsymbol{w}^{0}_{(k)}\right)\bigg\}.
\end{equation}
\: (2) \textbf{Player} $i$: Receive $\overline{\lambda}_{(k)}$ and update:
\begin{align}
    & \hspace{-0.7cm}\widehat{F}_{\mathcal{G}_{2}}^{i}(s_{0}, u_{(k)}, \boldsymbol{w}^{i}_{(k)}) = \frac{\sum_{l=1}^{M_{(k)}} F_{\mathcal{G}_{2}}^{i}(s_{0}, u_{(k)}, \SYR{\boldsymbol{w}^{i}_{(k),[l]}})}{M_{(k)}}, \label{eq: computing F hat i, at each iteration k}\\
    & \hspace{-0.7cm}\widehat{H}^{i}(s_{0}, u_{(k)}, \boldsymbol{w}^{i}_{(k)}) = \frac{\sum_{l=1}^{M_{(k)}} H^{i}(u^{i}_{(k)}, u^{-i}_{(k)}, \SYR{\boldsymbol{w}^{i}_{(k),[l]}})}{M_{(k)}},\label{eq: compute Lambda hat i at each iteration k}\\
  & \tilde{u}^i_{(k)} = (1-\delta)u^i_{(k)} + \delta \tilde{u}^i_{(k-1)},
  \label{eq: averaging step u that compute tilde u}\\
& u^i_{(k+1)} = \text{proj}_{\mathcal{D}^i} \bigg[ \tilde{u}^i_{(k)} - \alpha_{\left(k\right)}\bigg(\widehat{F}^i_{\mathcal{G}_{2}}\left(u^{i}_{(k)}, u^{-i}_{(k)}, \boldsymbol{w}^{i}_{(k)}\right) \nonumber\\
&\qquad\qquad\qquad\qquad +\widehat{H}^{i}\left(s_{0},u_{(k)},\boldsymbol{w}^{i}_{(k)}\right) \overline{\lambda}_{\left(k\right)}\bigg)\bigg].
\label{eq: projection step u}
\end{align}
\STATE \textbf{Output:} Strategies $u^i$ for each player and \SYR{the Lagrangian multiplier $\overline{\lambda}$}.
\end{algorithmic}
\label{alg: suggested algorithm}
\end{algorithm}
Common assumptions for using approximations \eqref{eq: approximation of Eg at each iteration}, \eqref{eq: computing F hat i, at each iteration k}, and \eqref{eq: compute Lambda hat i at each iteration k} include selecting an appropriate batch size sequence \( M_{(k)} \) \cite{iusem2017extragradient, franci2021stochastic}.
\begin{assumption}[\SYR{increasing batch size}]
\label{ass: batch size assumption}
The batch sizes \( (M_{(k)})_{k \geq 1} \) satisfy $M_{(k)} \geq c(k + k_{0})^{a+1}$,
for some constants \( c, a > 0 \), and \(k_{0}>1\).
\end{assumption}
\smallskip
From \textit{Assumption~\ref{ass: batch size assumption}}, \( 1/M_{(k)} \) is summable, a standard \mbox{assumption} for using variance reduction techniques to control stochastic errors \cite{iusem2017extragradient}. 
\subsection{Convergence analysis of Algorithm \ref{alg: suggested algorithm}}
\label{susec: Convergence Analysis of Algorithm}
In this subsection, we analyze the convergence of \mbox{Algorithm~\ref{alg: suggested algorithm}}. To this end, we derive an upper bound for the variance of the approximation error of \( \mathcal{A} \) in \textit{Proposition~\ref{prop: expectation e_{(k)}^{2} such that filtration}}. Utilizing this proposition along with a supporting lemma, we demonstrate that the proposed algorithm converges almost surely (a.s.) to the \vSGNE of \( \mathcal{G}_{2} \).

Let us define \( z := \text{col}(u, \overline{\lambda}) \) and \(\tilde{z}:=\text{col}(\tilde{u},\tilde{\lambda})\). Algorithm~\ref{alg: suggested algorithm} can be written in compact form as
\begin{subequations}
\label{eq: a compact form of the algorithm in operator space}
\begin{align}
    \tilde{z}_{(k)} &= (1 - \delta)z_{(k)} + \delta \tilde{z}_{(k-1)}, \label{eq: a compact form of the algorithm in operator space-averaging step}\\
    z_{(k+1)} &= (\text{Id} + \alpha_{\left(k\right)} \mathcal{B})^{-1}\left(\tilde{z}_{(k)} - \alpha_{\left(k\right)} \widehat{\mathcal{A}}(z_{(k)})\right),\label{eq: a compact form of the algorithm in operator space-projection step}
\end{align}
\end{subequations}
where $\text{Id}(.)$ denotes the identity operator. Since we are using an estimate of the operator $\mathcal{A}\left(z_{(k)}\right)$, we need to characterize the effect of the estimation error on the convergence of the algorithm. 
Hence, we define $\boldsymbol{e}_{(k)}:=\widehat{\mathcal{A}}\left(z_{(k)},w_{(k)}\right) - \mathcal{A}\left(z_{(k)}\right) $ as the estimation error on the extended operator. 
We can represent $\boldsymbol{e}_{(k)}$ in terms  of the estimation error in the elements of $\mathcal{A}\left(s_{0},z_{(k)}\right)$ which are
\begin{equation}\label{eq: different errors definition}
    \SYR{e_{i,(k)} = \sum_{l=1}^{M_{(k)}}\frac{e_{i,(k),[l]}}{M_{(k)}},\quad i = 1,2,3,\:\:\:\:k\in \mathbb{N},}
\end{equation}
\SYR{where
\(e_{1,(k),[l]} = F_{\mathcal{G}_{2}}(s_0,u_{(k)},w_{(k),[l]}) - \mathbb{F}_{\mathcal{G}_{2}}(s_0,u_{(k)})\),
\(e_{2,(k),[l]} = H(s_0,u_{(k)},w_{(k),[l]}) - \mathbb{H}(s_0,u_{(k)})\), and 
\(e_{3,(k),[l]} = g(s_0,u_{(k)},w_{(k),[l]}) - \mathbb{E}[g(s_0,u_{(k)},w)]\).}
By substituting \eqref{eq: different errors definition} into the definition of $\boldsymbol{e}_{(k)}$, we get
$\boldsymbol{e}_{(k)} = \mathrm{col}\left( e_{1,(k)} + e_{2,(k)} \overline{\lambda}_{(k)}, -e_{3,(k)}\right)$.
\\
Let  $\mathcal{F}^{w}_{(k)} := \sigma\big\{X_{0},(w_{(k')})_{k'=0}^{k-1}\big\}$, for all $k \in \mathbb{N}\backslash\{0\}$, and $\mathcal{F}^{w}_{(0)} := \sigma\{X_{0}\}$.
In the stochastic framework, there are usually assumptions on the expected value and the variances of stochastic errors $e_{i,(k),[l]}$ \cite{iusem2017extragradient,franci2021stochastic}.
\SYR{
\begin{assumption}
\label{ass: unbiased sampling}
    For all $i=1,2,3$, $k \in \mathbb{N}$ and $1 \leq l \leq M_{(k)}$, we assume a.s., $\mathbb{E}\left[e_{i,(k),[l]} \mid \mathcal{F}^{w}_{(k)}\right] = 0$.
\end{assumption}
}
\SYR{\begin{assumption} \label{ass: bounded of variances of errors}
    There exist some $\sigma_{1},\sigma_{2},\sigma_{3} >~ 0$, such that, for $i = 1,2,3$, $k \in \mathbb{N}$, and $l = 1,2,\cdots,M_{(k)}$,
    $\mathrm{sup}_{u \in \mathcal{U}_{\mathcal{G}_2}}\mathbb{E}\left[\left\|e_{i,(k),[l]}\right\|_{F}^{2}\right]~\le~ \sigma_i^2$.
\end{assumption}}
\SYR{
\begin{proposition}\label{Proposition: variance bound of each error}
 For all $k \in \mathbb{N}$ and $i \in \{1,2,3\}$,  if Assumption~\ref{ass: bounded of variances of errors} holds, we have $\mathbb{E}[\|e_{i,(k)}\|^{2} \mid \mathcal{F}_{(k)}] \leq 4 \frac{\sigma_{i}^{2}}{M_{(k)}}$.
\end{proposition}
\begin{proof}
    See Appendix.
\end{proof}
}
\begin{proposition} \label{prop: expectation e_{(k)}^{2} such that filtration}
    For all $k \in \mathbb{N}$, if Assumption~\ref{ass: bounded of variances of errors} holds, we have
   $
        \mathbb{E} \left[ \left\|  \boldsymbol{e}_{(k)} \right\|^{2} \mid  \mathcal{F}^{w}_{(k)}  \right]
        \leq 
        \frac{4}{M_{\left(k\right)}}\bigg[  \left( 1 + \left\| \overline{\lambda}_{(k)}  \right\|  \right) \sigma_{1}^{2}  +  \left(  \left\|  \overline{\lambda}_{(k)} \right\|^{2} + \left\|  \overline{\lambda}_{(k)} \right\|  \right) \sigma_{2}^{2} + \sigma_{3}^{2}    \bigg]$.
\end{proposition}
\begin{proof}
    Based on definition of $\boldsymbol{e}_{(k)}$, we have
\begin{align*}
        &\left\|  \boldsymbol{e}_{(k)}  \right\|^{2} = \boldsymbol{e}_{(k)}^{\top} \boldsymbol{e}_{(k)} = e_{1,(k)}^{\top} e_{1,(k)} + 
        \overline{\lambda}_{(k)}^{\top}e_{2,(k)}^{\top} e_{2,(k)} \overline{\lambda}_{(k)} \nonumber\\
        & +2 \overline{\lambda}_{(k)}^{\top}e_{2,(k)}^{\top}e_{1,(k)} + e_{3,(k)}^{\top}e_{3,(k)}
        \leq \left\| e_{1,(k)} \right\|^{2} +
        \nonumber \\ &
        \left\|  \overline{\lambda}_{(k)} \right\|^{2} \left\|  e_{2,(k)} \right\|^{2}  
        +2 \left\|  \overline{\lambda}_{(k)}  \right\|   \left\|  e_{2,(k)} \right\| \left\| e_{1,(k)} \right\|  
         + \left\|  e_{3,(k)} \right\|^{2},
    \end{align*}
    where the last inequality is obtained by applying the Cauchy-Schwarz inequality. 
    Moreover, since 2$\left\|  e_{2,(k)} \right\| \left\| e_{1,(k)} \right\| \leq  \left\| e_{1,(k)} \right\|^2+\left\| e_{2,(k)} \right\|^2$, we obtain 
    \(
        \left\|  \boldsymbol{e}_{(k)}  \right\|^{2} \leq 
         \left(  1+\left\|  \overline{\lambda}_{(k)}  \right\|  \right) \Big\{\left\|  e_{1,(k)} \right\|^{2} +\left\|  \overline{\lambda}_{(k)} \right\|   \left\|  e_{2,(k)} \right\|^{2} \Big\} 
         +  \left\| e_{3,(k)} \right\|^{2}\).
    Then, by taking conditional expectations from both sides of the last inequality, we have
\begin{equation*}
    \begin{aligned}
        &\mathbb{E} \left[  \left\| \boldsymbol{e}_{(k)}  \right\|^{2}  \mid \mathcal{F}^{w}_{(k)}  \right] \leq \left(  1+\left\| \overline{\lambda}_{(k)}  \right\| \right) \bigg\{\mathbb{E}\left[   \left\|  e_{1,(k)} \right\|^{2}  \mid \mathcal{F}^{w}_{(k)}  \right]  \nonumber \\
        & \qquad + \left\| \overline{\lambda}_{(k)}  \right\|   \mathbb{E}\left[  \left\|e_{2,(k)}  \right\|^{2}   \mid \mathcal{F}^{w}_{(k)}  \right] \bigg\}  + \mathbb{E}\left[  \left\|  e_{3,(k)}  \right\|^{2}   \mid \mathcal{F}^{w}_{(k)} \right].
    \end{aligned}
\end{equation*}
\SYR{Now, by applying \textit{Proposition~\ref{Proposition: variance bound of each error}} to the right-hand side of the above inequality, we obtain the proposition result.}
\end{proof}
Finally, we indicate how to choose the parameters of the algorithm. This is fundamental for convergence analysis and, in practice, for the convergence speed.
\begin{assumption}[\SYR{Averaging Parameter}]\label{ass: averaging parameter}
    The averaging parameter $\delta$ in \eqref{eq: a compact form of the algorithm in operator space} is such that 
    \(
        \frac{1}{\phi} \leq \delta < 1,
        \)
    where $\phi = \frac{1+\sqrt{5}}{2}$ is the golden ratio.
\end{assumption}
\SYR{\begin{assumption}[Diminishing Non-Summable Step size]
\label{ass: step size alpha}
    The sequence $\left(\alpha_{(k)}\right)_{k \in \mathbb{N}}$ is a positive decreasing sequence such that $\sum_{k=0}^{\infty}\alpha_{\left(k\right)} = \infty$ and  
    $\mathrm{lim}_{k \to \infty}\frac{\alpha_{(k)}}{\alpha_{(k-1)}}\neq 0$.
\end{assumption}}
\SYR{
\begin{remark} \label{remark: step size advatages and examples}
In \textit{Assumption~\ref{ass: step size alpha}}, unlike \cite{franci2021stochastic,malitsky2020golden}, our approach does not require prior knowledge of the Lipschitz constant  $\ell_{\mathcal{A}}$, nor the use of adaptive procedures to approximate it at each iteration, which are often computationally expensive. Some typical step sizes that satisfy the conditions in \textit{Assumption~\ref{ass: step size alpha}} include \(\alpha_{(k)} = c/\sqrt{k+a}\) and \(\alpha_{(k)} = c/(k+a)\), where \mbox{\(c,a > 0\)} are constant. 
\end{remark}
}
Let us define the residual function of $z_{(k)}$ as 
\(
\mathrm{res}\left(z_{\left( k \right)} \right) = \left\|   z_{\left( k \right)} - \left( \mathrm{Id}+\alpha_{\left(k\right)} \mathcal{B} \right)^{-1}\left( z_{\left( k \right)} - \alpha_{\left(k\right)} \mathcal{A}\left( z_{\left( k \right)} \right) \right)   \right\|.
\)
\begin{lemma}\label{lem: main lemma about inequality of residual}
     Let \textit{Assumption~\ref{ass: cost function in system model}--\ref{ass: step size alpha}} hold, and $\left( z_{(k)},\tilde{z}_{(k)} \right)_{k \in \mathbb{N}}$ be generated by Algorithm~\ref{alg: suggested algorithm}. Then, the following inequalities hold:   
\begin{equation}
     \label{eq: inequality residual and varpi_k and others 1}
	\left\| \tilde{z}_{\left( k \right)} \!-\! z_{\left( k \right)} \right\|^{2} \geq \frac{1}       {4}\mathrm{res}\left( z_{\left( k \right)} \right)^{2}\! -\! \frac{1}{2}\left\| \Delta z_{(k+1)}   \right\|^{2}
    \!-\! \alpha_{(k)}^{2} \left\|   e_{\left( k \right)} \right\|^{2},
\end{equation}
 \begin{align} \label{eq: main inequality third reformulation0}
 &\frac{2}{1-\delta}\left\| \Delta \tilde{z}_{(k+1)}^{\star} \right\|^{2} 
	+ \frac{\alpha_{(k)}}{\alpha_{(k-1)}}\frac{1}{\delta}\left\| \Delta z_{(k+1)} \right\|^{2} 
    \nonumber \\
	&
    \leq \frac{2}{1-\delta}\left\| \Delta \tilde{z}_{(k)}^{\star} \right\|^{2} -\frac{\alpha_{(k)}}{\alpha_{(k-1)}}\frac{2}{\delta}\left\| z_{\left( k \right)}-\tilde{z}_{\left( k \right)} \right\|^{2}  \nonumber\\
	& + 2\alpha_{(k)} \ell_{\mathcal{A}}\left( \left\| \Delta z_{(k)} \right\|^{2}+\left\| \Delta z_{(k+1)} \right\|^{2} \right)  \\
	& + 2\alpha_{(k)} \left( \left\| \Delta e_{(k)} \right\|^{2}+ \left\|  \Delta z_{(k+1)}   \right\|^{2} 
-2\big\langle     e_{\left( k \right)},\Delta z_{(k)}^{\star}   \big\rangle\right),\nonumber
\end{align}
    where $\Delta z_{(k+1)} = z_{(k+1)} - z_{(k)}$, $\Delta \tilde{z}_{(k)}^{\star} = \tilde{z}_{(k)}-z_{\star}$, $\Delta z_{(k)}^{\star} = z_{(k)}-z_{\star}$, and $\Delta e_{(k)} = e_{(k)} - e_{(k-1)}$.
\end{lemma}
\begin{proof}
    See Appendix.
\end{proof}
\begin{theorem}\label{theorem: convergence}
    Let \textit{Assumption~\ref{ass: cost function in system model}--\ref{ass: step size alpha}} hold. Then, the sequence $(u_{(k)})_{k \in \mathbb{N}}$ generated by Algorithm~\ref{alg: suggested algorithm}
    converges a.s. to a \vSGNE of $\mathcal{G}_{2}$.
    \end{theorem}
\begin{proof}
By applying the lower bound of $\left\| \tilde{z}_{(k)} - z_{(k)} \right\|^2$ from \eqref{eq: inequality residual and varpi_k and others 1} in \eqref{eq: main inequality third reformulation0}, we obtain:
\begin{align*}
 & \frac{2}{1-\delta}\left\| \Delta \tilde{z}_{\left( k+1 \right)}^{\star} \right\|^{2} + 4\alpha_{(k)}\left\langle  e_{\left( k \right)},\Delta z_{\left( k \right)}^{\star}  \right\rangle \nonumber \\
        &\quad +\left(  \frac{\alpha_{(k)}}{\alpha_{(k-1)}}\frac{1}{2\delta} -  2\alpha_{(k)}\left(\ell_{\mathcal{A}} +1 \right)\right)\left\| \Delta z_{\left( k+1 \right)} \right\|^{2} \nonumber \\
	&\leq \frac{2}{1-\delta}\left\| \Delta \tilde{z}_{\left( k \right)}^{\star} \right\|^{2}- \frac{\alpha_{(k)}}{\alpha_{(k-1)}}\frac{1}{4\delta}\mathrm{res}\left( z_{\left( k \right)} \right)^{2}  \nonumber \\
        &\quad  -\frac{\alpha_{(k)}}{\alpha_{(k-1)}}\frac{1}{\delta}\left\| \tilde{z}_{\left( k \right)} - z_{\left( k \right)} \right\|^{2} +\frac{\alpha_{(k)}}{\alpha_{(k-1)}}\frac{1}{\delta}\alpha_{(k)}^{2}\left\| e_{\left( k \right)} \right\|^{2}\nonumber \\
	&\quad  + 2\ell_{\mathcal{A}}\alpha_{(k)}\left\| \Delta z_{\left( k \right)} \right\|^{2} + 2\alpha_{(k)} \left\| \Delta e_{\left( k \right)} \right\|^{2} .
\end{align*}
Replacing $\left\| \Delta e_{\left( k \right)} \right\|^{2} $ with its upper bound $
     2\left\| e_{\left( k \right)} \right\|^{2}+ 2 \left\| e_{\left( k-1 \right)} \right\|^{2}$,
taking the expectation from both sides, and then using \textit{Proposition~\ref{prop: expectation e_{(k)}^{2} such that filtration}} and \textit{Assumption~\ref{ass: unbiased sampling}}
\begin{align}\label{eq: Fundumental inequality before applying lack of Lipschitz knowledge}
	&\mathbb{E}\bigg[   \frac{2}{1-\delta}\left\| \Delta \tilde{z}_{\left( k+1 \right)}^{\star} \right\|^{2}  \nonumber \\ 
    & +\left( \frac{\alpha_{(k)}}{\alpha_{(k-1)}}\frac{1}{2\delta} -2\alpha_{(k)}\left(\ell_{\mathcal{A}} +1 \right) \right) \left\| \Delta z_{\left( k+1 \right)} \right\|^{2}   \mid   \mathcal{F}^{w}_{\left( k \right)}\bigg]  \nonumber \\
        &+ \frac{\alpha_{(k)}}{\alpha_{(k-1)}}\frac{1}{4\delta}\mathrm{res}\left( z_{\left( k \right)} \right)^{2} + \frac{\alpha_{(k)}}{\alpha_{(k-1)}}\frac{1}{\delta}\left\|\tilde{z}_{(k)}- z_{(k)}\right\|^2  \nonumber \\
	&\leq \frac{2}{1-\delta}\left\|  \Delta \tilde{z}_{\left( k \right)}^{\star}  \right\|^{2} + 2\ell_{\mathcal{A}} \alpha_{(k)} \left\| \Delta z_{\left( k \right)} \right\|^{2} \nonumber \\
    & +\left(\frac{\alpha_{(k)}^{2}}{\delta}+4\alpha_{(k)} \right) \frac{4}{M_{(k)}} 
        \bigg\{ \left( 1+ \left\| \overline{\lambda}_{(k)} \right\| \right)\sigma_{1}^{2} \nonumber  \\
        & +\left( \left\| \overline{\lambda}_{(k)} \right\|^{2}+\left\| \overline{\lambda}_{(k)} \right\| \right) \sigma_{2}^{2}+  \sigma_{3}^{2} \bigg\} \nonumber \\
        & + \frac{16\, \alpha_{(k)}}{M_{(k-1)}} 
        \bigg\{ \left( 1+ \left\| \overline{\lambda}_{(k-1)} \right\| \right)\sigma_{1}^{2}   \nonumber \\
        & +\left( \left\| \overline{\lambda}_{(k-1)} \right\|^{2}+\left\| \overline{\lambda}_{(k-1)} \right\| \right) \sigma_{2}^{2}+  \sigma_{3}^{2} \bigg\}.
\end{align}
\SYR{
Since $\alpha_{(k)} \to 0$, there exist $k_{1} \geq 0$, such that $\alpha_{(k)} \leq \frac{1}{2\, \delta \, (4 \ell_{\mathcal{A}}+2)}$, for all $k \geq k_{1}$. Thus, we have
\(\frac{1}{2\, \delta} - 2\alpha_{(k)}(\ell_{\mathcal{A}}+1) \geq 2 \ell_{\mathcal{A}}\, \alpha_{(k)}, \forall k \ge k_{1}
\).} Moreover, using \textit{Assumptions~\ref{ass: averaging parameter}--\ref{ass: step size alpha}} and \textit{Assumption~\ref{ass: second assumption about cost function}} it can be easily shown that $\bar{\lambda}_{(k)}\leq \bar{\lambda}_{\mathrm{max}} < \infty$.
Now
using 
\textit{Proposition~\ref{prop: expectation e_{(k)}^{2} such that filtration}} and \eqref{eq: Fundumental inequality before applying lack of Lipschitz knowledge}
yields 
\begin{equation}\label{eq: inequality match to Robbins Siegmund lemma}
    \mathbb{E}\left[ v_{(k+1)} \mid \mathcal{F}_{(k)}^{w} \right] + \theta_{(k)} 
        \leq
         v_{(k)} + \eta_{(k)},\hspace{0.5cm} \forall \SYR{k\geq k_{1}}
\end{equation}
where,
\begin{equation*}
\begin{aligned}
	v_{\left( k \right)}&=\frac{2}{1-\delta}\left\| \Delta \tilde{z}_{\left( k \right)}^{\star} \right\|^{2}+ 2\ell_{\mathcal{A}}\alpha_{(k)} \left\| \Delta z_{\left( k \right)}\right\|^{2}, \nonumber \\
	\theta_{\left( k \right)} &= \frac{\alpha_{(k)}}{\alpha_{(k-1)}}\left[\frac{1}{4 \delta}\mathrm{res}\left( z_{\left( k \right)} \right)^{2} + \frac{1}{\delta}\left\|\tilde{z}_{(k)}- z_{(k)}\right\|^2\right],  \nonumber \\
	\eta_{\left( k \right)} &=  \left(\frac{\alpha_{(k)}^{2}}{\delta}+4\alpha_{(k)} \right) \frac{4}{M_{(k)}} 
    \Big\{ \left( 1+ \overline{\lambda}_{\mathrm{max}} \right)\sigma_{1}^{2} \nonumber \\
    & +\left( \overline{\lambda}_{\mathrm{max}}^{2}+\overline{\lambda}_{\mathrm{max}} \right) \sigma_{2}^{2}+  \sigma_{3}^{2} \Big\}  + \frac{16\, \alpha_{(k)}}{M_{(k-1)}} 
    \Big\{ \left( 1+ \overline{\lambda}_{\mathrm{max}}\right)\sigma_{1}^{2} \nonumber \\
    & +\left( \overline{\lambda}_{\mathrm{max}}^{2}+\overline{\lambda}_{\mathrm{max}} \right) \sigma_{2}^{2}+  \sigma_{3}^{2} \Big\}.
\end{aligned}
\end{equation*}
Based on \textit{Assumption~\ref{ass: step size alpha}}, and \textit{Assumption~\ref{ass: batch size assumption}}, it can be easily shown that $\sum_{k = 0}^{\infty} \eta_{(k)} < \infty$.
Now, using the Robbins-Siegmund lemma (see \cite{robbins1971convergence}), we conclude that \(v_{(k)}\) converges and that \(\sum_{k} \theta_{(k)} < \infty\). Given that \(\{\theta_{(k)}\}\) is non-negative and summable, it follows that \(\lim_{k \to \infty} \theta_{(k)} = 0\). Consequently, it implies \(\tilde{z}_{(k)} \to z_{(k)}\) and \SYR{\(\text{res}\left(z_{(k)}\right)^2 \to 0\)}. Moreover, we can find that \SYR{$\text{res}\left(\tilde{z}_{(k)}\right)^2 \to 0$}, and the sequence \(\{\tilde{z}_{(k)}\}\) is bounded.
\end{proof}
\section{Simulation Results}\label{sec: simulation}
\SYG{We study a grid-connected community microgrid with $N$ residential players over a horizon $T$ and a time step $\Delta t$. The shared battery state of charge (SoC), $SoC_t$, evolves according to
$
\mathrm{SoC}_{t+1}=\mathrm{SoC}_t+\eta\,\Delta t\left(r_t-\sum_{j=1}^N u^j_t\right)
$, where $u^j_t$ is household $j$'s discharge decision, $r_t$ is stochastic renewable generation, and $\eta\in(0,1]$ is the charge/discharge efficiency. Household $i$'s grid exchange is $g^i_t:=d^i_t-u^i_t$ for known demand $d^i_t$. Let the aggregate grid exchange be $G_t:=\sum_{j=1}^N g^j_t$ and the tariff be $
\pi_t \;=\; K^{\mathrm{ToU}}_t \;+\; \frac{k_c}{N}\,G_t$, with $K^{\mathrm{ToU}}_t$ the conventional time-of-use price and $k_c$ a congestion coefficient. Each player $i$ chooses $u^i=\{u^i_t\}_{t=0}^{T-1}$ to minimize  $\mathbb{J}^i = \mathbb{E} \big\{ \sum_{t=0}^{T-1} \big[ \pi(g_t)g_t^i + \sum_{j=1}^{N}(\alpha^{\mathrm{dgr}}(u_t^j)^2 + \beta^{\mathrm{dgr}}u_t^j) \big] - \alpha^{\mathrm{utl}}\ln (1+\sum_{t=0}^{T-1}g_t^i) + \frac{\alpha^{\mathrm{trm}}}{2}\|\mathrm{SoC}_T - \mathrm{SoC}^{\mathrm{des}}\|^2 \big\}$
where $\alpha^{\mathrm{dgr}},\beta^{\mathrm{dgr}}$ weight battery degradation, $\alpha^{\mathrm{utl}}$ weights user utility, and $\alpha^{\mathrm{trm}}$ weights the terminal SoC tracking toward the desired level $\mathrm{SoC}^{\mathrm{des}}$. Decisions satisfy local bounds $0\le u^i_t\le d^i_t$ and shared chance constraints
$
\mathbb{P}\{\mathrm{SoC}^{\min}\le \mathrm{SoC}^t\le \mathrm{SoC}^{\max}\}\ge 1-\hat{\gamma}$ and 
$\mathbb{P}\{|\mathrm{SoC}_T-\mathrm{SoC}^{\mathrm{des}}|\le c\}\ge 1-\tilde{\gamma},
$
with $\mathrm{SoC}^{\min},\mathrm{SoC}^{\max}$ the operating bounds, $c$ a terminal tolerance, and $\hat{\gamma},\tilde{\gamma}$ the violation probabilities.
The parameters are set as \(N=20\), \(T = 24\), \(\Delta t=1\), \(\eta = 5\times 10^{-5}\), \(\hat{\gamma} = 0.1\), \(\tilde{\gamma}= 0.1\), \(k_c = 1\), \(\alpha^{\mathrm{dgr}}= 80\), \(\beta^{\mathrm{dgr}}= 10\), \(\alpha^{\mathrm{utl}}= 50\), \(\alpha^{\mathrm{trm}} = 1\), \(\mathrm{SoC}_{0} = 0.1\), \(\mathrm{SoC}^{\mathrm{min}} = 0.1\), \(\mathrm{SoC}^{\mathrm{max}} = 0.9\), \(\mathrm{SoC}^{\mathrm{des}} = 0.5\), \(c = 0.05\). The demand profile for the players is adopted from \cite{yadollahi2023generalized}. We use the conventional ToU values $K^{\mathrm{ToU}}_t=\{15.3,\,35.6,\,23.3,\,45.6,\,27.6\}$ over the bands $[0\!-\!4], [5\!-\!14], [15\!-\!16], [17\!-\!21], [22\!-\!24]$. The batch size sequence is \(M_{(k)} = \lceil (k+2)^{1.1} \rceil\). 
Algorithm parameters are \(\delta = 0.9\) and \(\alpha_{(k)} =  10^{-4} /(k+2)\).}

\SYR{Figure~\ref{fig:convergence u} shows the convergence of \(u^{i}_{21}\) and \(u^{i}_{16}\) with 95\% confidence bands from 100 independent Monte Carlo simulations. The bands generally narrow as the algorithm advances, with transient widening early on, indicating almost sure convergence to the SGNE of the under-approximated game as iterations approach infinity.\\
Figure~\ref{fig:res} illustrates the convergence of \(\mathbb{E}[\mathrm{res}(z_{(k)})]\) and \(\mathbb{E}[\mathrm{res}^2(z_{(k)})]\) towards zero, although this convergence is not strictly monotonic.
Figure~\ref{fig:res} also shows convergence with fixed step size $10^{-4}$ but divergence at $0.01$, revealing that fixed step-size schemes converge only if the step size stays below a problem-dependent threshold. Excessively small steps can slow convergence, especially with large Lipschitz constants.
Despite divergence at $\alpha_{(k)}=0.01$, Figure~\ref{fig:res} confirms convergence with diminishing step sizes such as $\alpha_{(k)}=0.01/(k+2)$ and $\alpha_{(k)}=0.01/\sqrt{k+2}$, enabling larger steps early and progressively smaller ones near the solution.}
\begin{figure}
    \centering
    \includegraphics[height= 0.6\columnwidth ,width=\columnwidth]{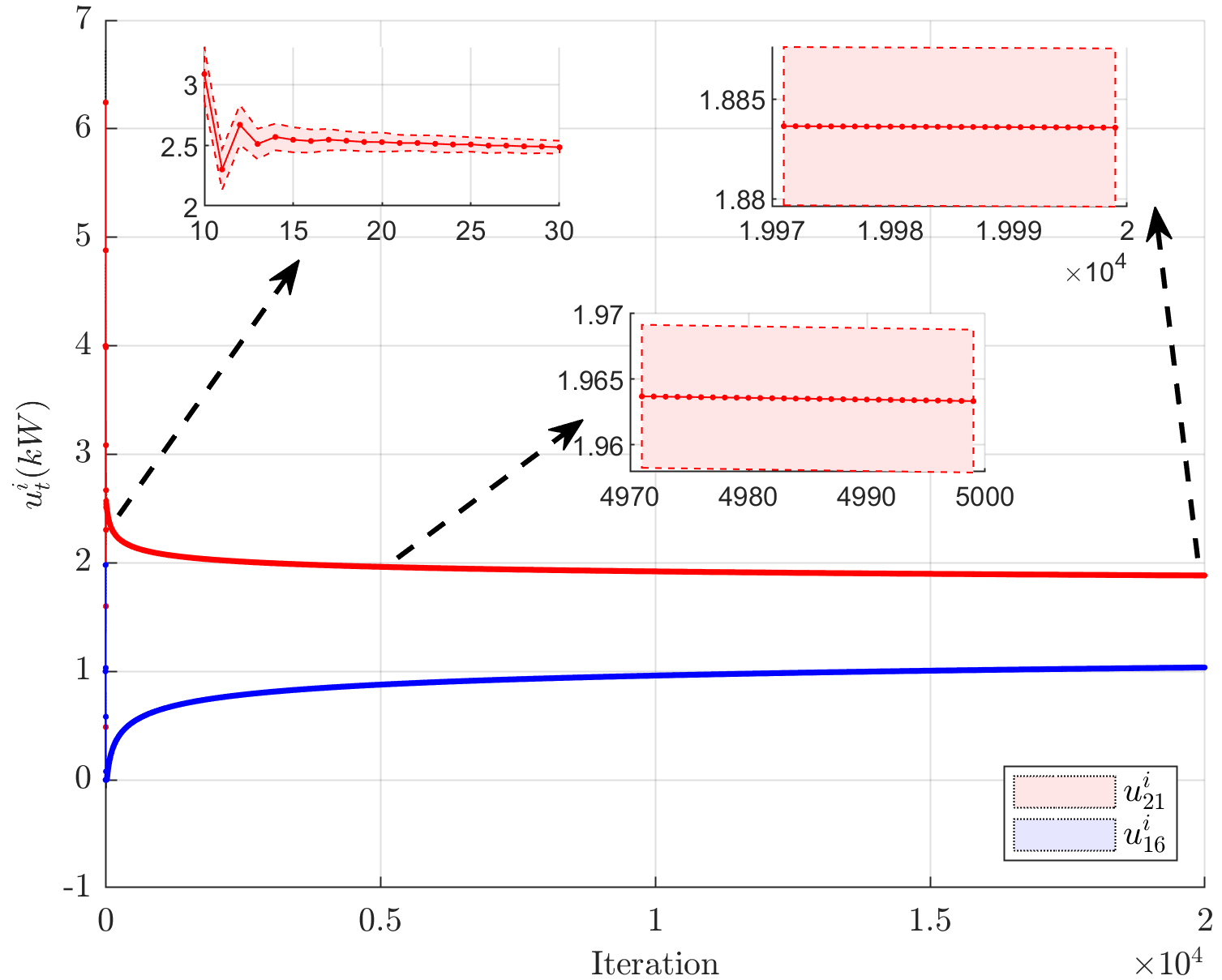}
    \caption{A visual representation of the  convergence of $u^{i}_{16},u^{i}_{21}$ during the execution of Algorithm~\ref{alg: suggested algorithm}.}
    \label{fig:convergence u}
\end{figure}
\begin{figure}
    \centering
    \includegraphics[height= 0.6\columnwidth ,width=\columnwidth]{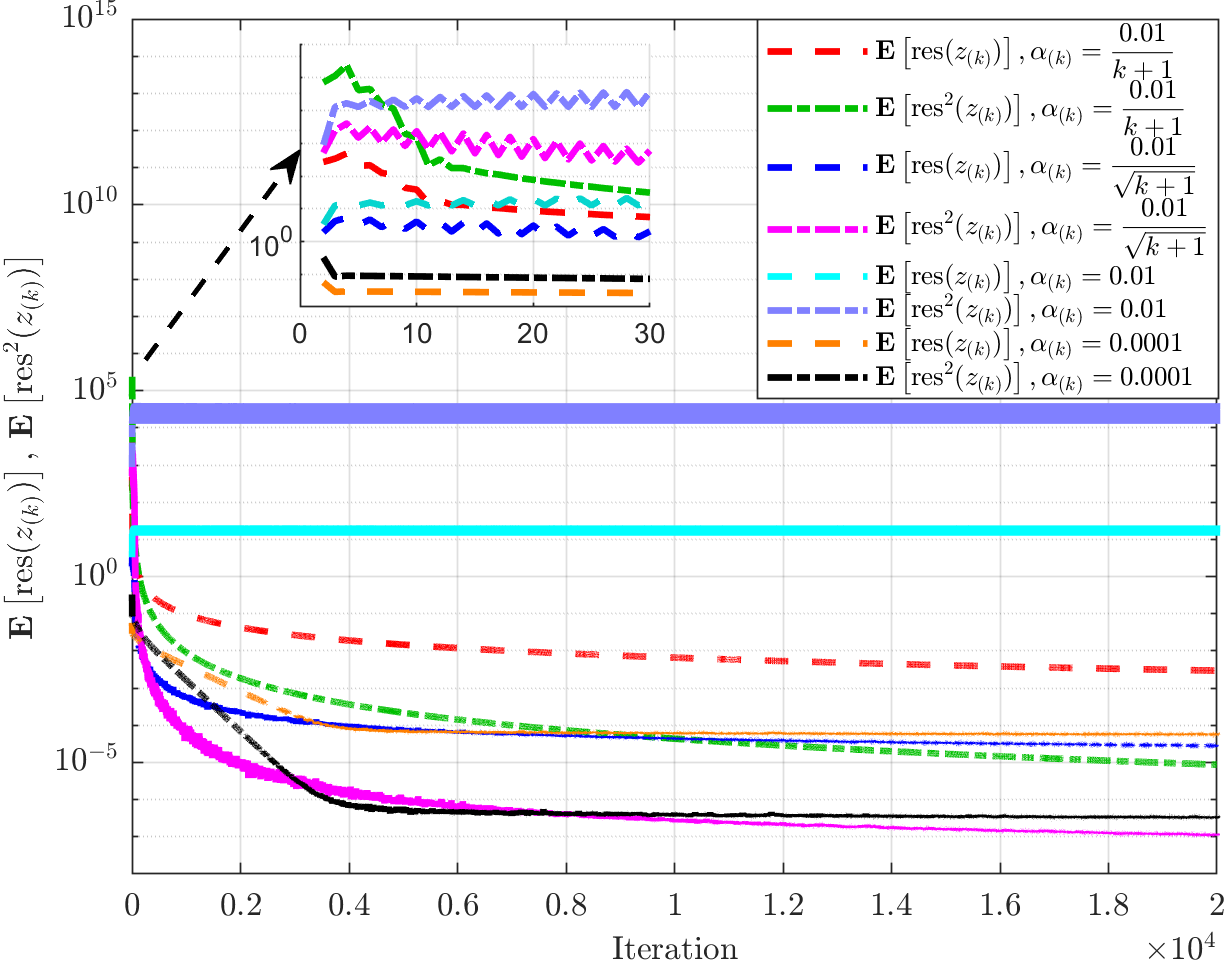}
    \caption{A visual representation of the  convergence of $\mathrm{res}(z)$ and $\mathrm{res}^{2}(z)$ during the execution of Algorithm~\ref{alg: suggested algorithm}.}
    \label{fig:res}
\end{figure}
\section{Conclusion}
\label{sec:conculsion}
This paper studied stochastic generalized dynamic games (SGDGs) with coupling chance constraints under uncertainties satisfying a concentration of measure property. By under-approximating these constraints via expectation constraints, we established the existence of stochastic generalized Nash equilibria (SGNEs) and showed that any variational SGNE of the approximated game is an $\varepsilon$-SGNE of the original game. We proposed an SGNE-seeking algorithm for the approximated game under mere monotonicity and Lipschitz continuity of the pseudo-gradient, proving almost-sure convergence. \SYR{The study of stochastic generalized nonconvex nonlinear dynamic games under shared chance constraints (with unknown random variables) is an important direction for future research.}
\appendix
\subsection*{Proof of lemma~\ref{lem: main lemma about inequality of residual}}
\SYR{The proof is inspired by \cite{malitsky2020golden,franci2021stochastic}.}\\
(1) Let $R_k:=(\operatorname{Id}+\alpha_{(k)}\mathcal{B})^{-1}$.
By firm nonexpansiveness of $R_k$ and the update $z_{(k{+}1)}=R_k\bigl(\tilde{z}_{(k)}-\alpha_{(k)}\hat A(z_{k)})\bigr)$,
\begin{align*}
&\operatorname{res}(z_{(k)})
=\bigl\|z_{(k)}-R_k(z_{(k)}-\alpha_{(k)}A(z_{(k))})\bigr\|
\\ &\le\|z_{(k)}-z_{(k{+}1})\|+\|z_{(k{+}1)}-R_k(z_{(k)}-\alpha_{(k)}A(z_{(k))})\|.
\end{align*}
Inserting the intermediate point $R_k(\tilde{z}_{(k)}-\alpha(k)A(z(k)))$ and using nonexpansiveness of $R_k$ twice gives
$$
\|z_{(k{+}1)}-R_k(z_{(k)}-\alpha_{(k)}A(z_{(k)}))\|
\le \|\tilde{z}_{(k)}-z_{(k)}\|+\alpha_{(k)}\|e_{(k)}\|.
$$
Thus
$
\operatorname{res}(z_{(k)})
\le \|\Delta z_{(k{+}1)}\|+\|\tilde{z}_{(k)}-z_{(k)}\|+\alpha_{(k)}\|e_{(k)}\|.
$
Finally, by applying $\|a+b\|^2\le 2\|a\|^2+2\|b\|^2$ twice to the above inequality and then rearranging, we obtain \eqref{eq: inequality residual and varpi_k and others 1}.\\
(2) 
Let \( Q(u, \overline{\lambda}) = \sum_{i=1}^{N} \iota_{\mathcal{D}^{i}}(u^{i}) + \sum_{j=1}^{m} \iota_{\mathbb{R} \geq 0}(\overline{\lambda}^{j}) \), where \(\overline{\lambda}^{j}\) denotes the \(j^{\text{th}}\) entry of \(\overline{\lambda}\). Since $\mathcal{B}=\partial Q$, the resolvent–prox identity gives $R_{k}^{-1}=\mathrm{prox}_{\alpha_{(k)}Q}$. 
Applying prox-inequality \cite[Proposition 12.26]{bauschke2011} at iteration $k$ to update $z_{(k+1)}=\mathrm{prox}_{\alpha_{(k)}Q}\bigl(\tilde{z}_{(k)}-\alpha_{(k)}\hat A(z_{(k))}\bigr)$ 
,yields
\begin{align}
\label{eq: using proximal equivalence1 in algorithm-1}
-\left\langle z_{(k+1)} - \tilde{z}_{(k)} + \alpha_{(k)} 	\hat{\mathcal{A}}\left( z_{(k)}\right), \Delta z_{(k+1)}^{\star}   \right\rangle \nonumber\\ \geq \alpha_{(k)} \left(  Q(z_{(k+1)}) - Q(z_{\star}) \right).  
\end{align}
Applying the same argument to the update at iteration $k-1$
and using the identity $z_{(k)} - \tilde{z}_{(k-1)} = \frac{1}{\delta}\left( z_{(k)}-\tilde{z}_{(k)}  \right)$, which follows from \eqref{eq: a compact form of the algorithm in operator space-averaging step}, obtain
\begin{align}
\label{eq: equation using lemma 6 and 33}
    \left \langle  \frac{1}{\delta}\left(  z_{(k)} - \tilde{z}_{(k)} \right) + \alpha_{(k-1)} \hat{\mathcal{A}}\left( z_{(k-1)}   \right) , \Delta z_{(k+1)} \right \rangle \nonumber \\ \geq \alpha_{(k-1)} \left( Q\left( z_{(k)} \right)  - Q\left(z_{(k+1)}\right)  \right).
\end{align} 
Multiplying \eqref{eq: equation using lemma 6 and 33} by $\alpha_{(k)}/\alpha_{(k-1)}$ and adding to \eqref{eq: using proximal equivalence1 in algorithm-1}, and then expanding
$\widehat A(z_{(k)})=A(z_{(k)})+e_{(k)}$ together with
$\langle a,b\rangle=\tfrac12(\|a\|^{2}+\|b\|^{2}-\|a-b\|^{2})$,
one obtains after straightforward algebra
\begin{align}
\label{eq: main inequality first reformulation}
	&-\left\|      z_{\left( k+1 \right)}- \tilde{z}_{\left( k \right)} \right\|^{2} - 
	\left\| \Delta z_{\left( k+1 \right)}^{\star}\right\|^{2}  + \left\| \Delta \tilde{z}_{\left( k \right)}^{\star} \right\|^{2} -\frac{\alpha_{(k)}}{\alpha_{(k-1)}}\frac{1}{\delta}\nonumber \\
	& 
	\left\{\left\| z_{\left( k \right)} - \tilde{z}_{\left( k \right)} \right\|^{2} + \left\| \Delta z_{\left( k+1 \right)} \right\|^{2}
	- \left\| z_{\left( k+1 \right)} - \tilde{z}_{\left( k \right)}\right\|^{2}\right\}  \nonumber \\
	&- 2\alpha_{(k)}\bigg\{\left\langle \mathcal{A}\left( z_{\left( k \right)} \right),\Delta z_{\left( k \right)}^{\star} \right\rangle + \left\langle  \boldsymbol{e}_{\left( k \right)},\Delta z_{\left( k \right)}^{\star} \right\rangle \nonumber \\
	&+\left\langle  \mathcal{A}\left( z_{\left( k \right)}\right) - \mathcal{A}\left( z_{\left( k-1 \right)} \right)  ,\Delta z_{\left( k+1 \right)} \right\rangle  +\left\langle  \Delta \boldsymbol{e}_{\left( k \right)} ,\Delta z_{\left( k+1 \right)} \right\rangle\bigg\}\nonumber \\
	&
	\geq 2 \alpha_{(k)} \left( Q\left( z_{\left( k \right)} \right)-Q\left( z_{\star} \right) \right) .
\end{align}
From \eqref{eq: a compact form of the algorithm in operator space-averaging step} we have
    \begin{equation}
    \label{eq: first equation of algorithm in compact form after substituting k with k+1}
        \tilde{z}_{(k+1)} = \left( 1-\delta \right) z_{(k+1)} + \delta \tilde{z}_{(k)}. 
    \end{equation}
    Solving $z_{(k+1)}$ from \eqref{eq: first equation of algorithm in compact form after substituting k with k+1} and subtracting \(z_{\star}\) yields
    $ \Delta z_{(k+1)}^{\star} = \frac{1}{1-\delta}  \Delta \tilde{z}_{(k+1)}^{\star}  - \frac{\delta}{1-\delta} \Delta \tilde{z}_{(k)}^{\star}$.
Taking squared norms and expanding using cosine rule, gives
\begin{align}
\label{eq: norm delta zstark+1power2-1}
\left\|\Delta z_{(k+1)}^{\star}\right\|^{2}&=\frac{1}{(1-\delta)^{2}}\left\|\Delta \tilde{z}_{(k+1)}^{\star}\right\|^{2}+\frac{\delta ^{2}}{(1-\delta)^{2}}\left\| \Delta \tilde{z}_{(k)}^{\star}\right\|^{2} \nonumber \\&- 2 \frac{\delta}{(1-\delta)^{2}}\left\langle \Delta \tilde{z}_{(k+1)}^{\star},\Delta \tilde{z}_{(k)}^{\star}\right\rangle.
\end{align}
Applying \eqref{eq: first equation of algorithm in compact form after substituting k with k+1} to the last term of \eqref{eq: norm delta zstark+1power2-1} and rearranging  yields 
\begin{align}
\label{eq: eq jadid before eq instead of lemma3-3}
	&\left\| \Delta z_{\left( k+1 \right)}^{\star} \right\|^{2} 
        = \frac{1}{\left( 1-\delta \right)^{2}}\left\| \Delta \tilde{z}_{\left( k+1 \right)}^{\star} \right\|^{2}+\frac{\delta^{2}}{\left( 1-\delta \right)^{2}}\left\| \Delta \tilde{z}_{\left( k \right)}^{\star}\right\|^{2} \nonumber \\
	&-\frac{2\delta}{1-\delta}\left\langle   z_{\left( k+1 \right)}-\tilde{z}_{\left( k \right)}  , \Delta \tilde{z}_{\left( k \right)}^{\star}\right\rangle 
	 -\frac{2\delta}{\left( 1-\delta \right)^{2}}\left\| \Delta\tilde{z}_{\left( k \right)}^{\star} \right\|^{2}
\end{align}Moreover, using \eqref{eq: first equation of algorithm in compact form after substituting k with k+1}, we have 
    $
         \left\|  \Delta \tilde{z}_{(k+1)}^{\star}    \right\|^{2} 
        = \left\|  \left( 1-\delta \right)z_{(k+1)} + \delta \tilde{z}_{(k)} - z_{\star}   \right\|^{2} 
         =~\left( 1-\delta \right)^{2}  \left\|  z_{(k+1)} - \tilde{z}_{(k)}  \right\|^{2} + \left\| \Delta \tilde{z}_{(k)}^{\star} \right\|^{2} 
         + 2\left( 1-\delta \right) \langle  z_{(k+1)} - \tilde{z}_{(k)} , \Delta \tilde{z}_{(k)}^{\star} \rangle . 
    $
    Multiplying both sides by \(\frac{-\delta}{\left( 1 - \delta \right)^{2}}\) and rearranging gives
\begin{align}
\label{eq: instead of lemma3-3}
&\frac{- \delta}{\left(  1-\delta \right)^{2}}\left\| \Delta \tilde{z}_{\left( k+1 \right)}^{\star} \right\|^{2} + \delta\left\| z_{(k+1)}-\tilde{z}_{(k)} \right\|^{2} + \frac{\delta}{\left( 1-\delta \right)^{2}}\left\| \Delta \tilde{z}_{(k)}^{\star}  \right\|^{2}\nonumber\\ &= -2 \frac{\delta}{1-\delta}\langle  z_{(k+1)} - \tilde{z}_{(k)},\Delta \tilde{z}_{(k)}^{\star}  \rangle\end{align}
Now, combining \eqref{eq: eq jadid before eq instead of lemma3-3} and \eqref{eq: instead of lemma3-3} \SY{yields}
\begin{align}
\label{eq: instead of lemma 3(7)}
    \left\| \Delta z_{\left( k+1 \right)}^{\star} \right\|^{2}
    = \frac{1}{1-\delta}\left\| \Delta \tilde{z}_{\left( k+1 \right)}^{\star} \right\|^{2} - \frac{\delta}{1-\delta}\left\| \Delta \tilde{z}_{\left( k \right)}^{\star} \right\|^{2} \nonumber \\+ \delta \left\| z_{\left( k+1 \right)}-\tilde{z}_{\left( k \right)} \right\|^{2}.
\end{align}
Substituting $\left\|\Delta z_{(k+1)}^{\star} \right\|^2$ from \eqref{eq: instead of lemma 3(7)} into \eqref{eq: main inequality first reformulation} yields
\begin{align}
\label{eq: main inequality second reformulation}
	& \frac{1}{1-\delta}\left\| \Delta \tilde{z}_{\left( k+1 \right)}^{\star} \right\|^{2} + 
	\frac{\alpha_{(k)}}{\alpha_{(k-1)}}\frac{1}{\delta}\left\| \Delta z_{\left( k+1 \right)} \right\|^{2} \nonumber \\
	&\quad +2\alpha_{(k)} \left( Q\left( z_{\left( k \right)} \right)-Q\left( z_{\star} \right) \right) 
        +2\alpha_{(k)}\left\langle  \mathcal{A}\left( z_{\left( k \right)} \right),\Delta z_{\left( k \right)}^{\star}     \right\rangle \nonumber\\
	&\leq 
	\left( \frac{\delta}{1-\delta} \right)\left\|    \Delta \tilde{z}_{\left( k \right)}^{\star} \right\|^2 - \frac{\alpha_{(k)}}{\alpha_{(k-1)}}\frac{1}{\delta}\left\|   z_{\left( k \right)}-\tilde{z}_{\left( k \right)}   \right\|^2\nonumber \\
     	&\quad +\left( -1-\delta+\frac{\alpha_{(k)}}{\alpha_{(k-1)}}\frac{1}{\delta} \right)\left\|    z_{\left( k+1 \right)} -\tilde{z}_{\left( k \right)}  \right\|^{2} \nonumber \\
	&\quad -2 \alpha_{(k)}\left\langle    \mathcal{A}\left( z_{\left( k \right)}\right) -\mathcal{A}\left( z_{\left( k-1 \right)}\right),\Delta z_{(k+1)}     \right\rangle \nonumber \\
	&\quad - 2\alpha_{(k)}\left\langle    e_{\left( k \right)},\Delta z_{\left( k \right)}^{\star}    \right\rangle 
        -2\alpha_{(k)}\left\langle   \Delta e_{\left( k \right)}    ,\Delta z_{\left( k+1 \right)}        \right\rangle .  
\end{align}
Since $\mathcal{A}$ is monotone (\textit{Lemma~\ref{lem: equivalence of answer distributed and centralized algorithm (KKT)}}(\ref{lem: monotonicity of A and Lipschitz})), and $Q\left(z_{(k)}\right) \geq Q\left(z_{\star}\right)$,
\begin{align}
\label{eq: eq resulted from monotonicity of A}
	&\left\langle \mathcal{A}\left( z_{\left( k \right)} \right),\Delta z_{\left( k \right)}^{\star} \right\rangle 
	\geq 
	\left\langle        \mathcal{A}\left( z_{\star} \right),\Delta z_{\left( k \right)}^{\star}         \right\rangle \iff \nonumber \\
	&\left\langle \mathcal{A}\left( z_{\left( k \right)} \right),\Delta z_{\left( k \right)}^{\star} \right\rangle +  \left( Q\left( z_{\left( k \right)}\right) - Q\left( z_{\star} \right) \right) 
	\geq \nonumber \\
	&\quad \geq\left\langle        \mathcal{A}\left( z_{\star} \right),\Delta z_{\left( k \right)}^{\star}         \right\rangle +  \left( Q\left( z_{\left( k \right)}\right) - Q\left( z_{\star} \right) \right)	\geq 0.
\end{align}
Given \textit{Assumption~\ref{ass: averaging parameter}} and \textit{Assumption~\ref{ass: step size alpha}}, it can be easily shown that 
$\left(-1-\delta+ \frac{\alpha_{(k)}}{\alpha_{(k-1)}}\frac{1}{\delta}\right)$ is non-positive.
Now, based on this fact and using \eqref{eq: main inequality second reformulation} and \eqref{eq: eq resulted from monotonicity of A}, we can find
\begin{align}
\label{eq: main inequality third reformulation}
    & \frac{1}{1-\delta}\left\| \Delta \tilde{z}_{\left( k+1 \right)}^{\star} \right\|^{2} + 
	\frac{\alpha_{(k)}}{\alpha_{(k-1)}}\frac{1}{\delta}\left\| \Delta z_{\left( k+1 \right)} \right\|^{2} \nonumber \\
	&\leq 
	\frac{1}{1-\delta} \left\|   \Delta \tilde{z}_{\left( k \right)}^{\star} \right\|^2 - \frac{\alpha_{(k)}}{\alpha_{(k-1)}}\frac{1}{\delta}\left\|   z_{\left( k \right)}-\tilde{z}_{\left( k \right)}   \right\|^2 \nonumber \\
	&\quad -2  \alpha_{(k)}\left\langle \mathcal{A}\left( z_{\left( k \right)}\right) -\mathcal{A}\left( z_{\left( k-1 \right)}\right), \Delta z_{(k+1)}     \right\rangle \nonumber \\
	&\quad -2 \alpha_{(k)}\left\langle       e_{\left( k \right)},\Delta z_{\left( k \right)}^{\star}    \right\rangle 
        - 2\alpha_{(k)}\left\langle   \Delta e_{\left( k \right)} , \Delta z_{\left( k+1 \right)}        \right\rangle.
\end{align}
By using Lipschitz continuity of $\mathcal{A}$ (\textit{Lemma~\ref{lem: equivalence of answer distributed and centralized algorithm (KKT)}} (\ref{lem: monotonicity of A and Lipschitz})), Cauchy-Schwartz and Young's inequality, we obtain that
$-\left\langle    \mathcal{A}\left( z_{\left( k \right)} \right)-\mathcal{A}\left( z_{\left( k-1 \right)}   \right) , \Delta z_{(k+1)}     \right\rangle             
    \leq 
    \frac{ \ell_{\mathcal{A}}}{2} \left( \left\|      \Delta z_{\left( k \right)}       \right\|^{2}    +     \left\|      \Delta z_{\left( k+1 \right)}       \right\|^{2} \right)$
and
	$- \left\langle     \Delta e_{\left( k \right)},\Delta z_{\left( k+1 \right)}   \right\rangle 
        \leq     
      \frac{1}{2}\left( \left\| \Delta e_{(k)} \right\|^2 + \left\| \Delta z_{\left( k+1 \right)}   \right\|^{2} \right)$ .
Substituting  into \eqref{eq: main inequality third reformulation}, and then 
Multiplying both sides by $2$ and subtracting the left-hand side of the result from $\tfrac{\alpha_{(k)}}{\alpha_{(k-1)}}\tfrac{1}{\delta}\|\Delta z_{(k+1)}\|^{2}$ yields \eqref{eq: main inequality third reformulation0}.
\subsection*{Proof of Proposition~\ref{Proposition: variance bound of each error}}
\SYR{
Fix $k \in \mathbb{N}$. For $i\in\{1,2,3\}$ define the processes
\(
d_{i,(k),[m']} := 
\sum_{l=1}^{m'} \frac{e_{i,(k),[l]}}{M_{(k)}}\), for \(1\leq m' \leq M_{(k)}\), and $d_{i,(k),[0]} = 0$. And, define  $\mathcal{F}^{\mathrm{d}}_{(k),[m']} = \sigma(w_{(k)[1]},w_{(k),[2]},\cdots,w_{(k),[m']}) $ .
Then $\{d_{i,(k),[m']}(u),\widehat{\mathcal{F}}_{(k),[m']}\}_{m'=0}^{M_{(k)}}$ is a martingale with increments
\(
\Delta d_{i,(k),[m']}:= d_{i,(k),[m']}-d_{i,(k),[m'-1]}
= \frac{e_{i,(k),[m']}}{M_{(k)}}
\).
By \textit{Assumption~\ref{ass: bounded of variances of errors}},
\(
\mathbb{E}\bigl[\|\Delta d_{i,(k),[m']}\|^{2}\bigr]
   \le \frac{4\sigma_i^2}{M_{(k)}^{2}}
\).
Since $d_{i,(k),[M_{(k)}]}=e_{i,(k)}$, applying Doob’s Inequality \cite[Section 4.4]{durrett2019probability} and some norms property, such as $\|\cdot\|\le \|\cdot\|_{F}$, yields the claimed bound, which proves the proposition.
}
\section*{References}
\bibliographystyle{ieeetr}
\bibliography{main}
\newpage
\addtolength{\textheight}{-12cm} 
\end{document}